\documentclass[letterpaper,11pt]{article}
\usepackage[margin=1in]{geometry}
\usepackage{amsmath,amssymb,amsthm}
\usepackage[table]{xcolor}
\usepackage{mathtools}
\usepackage{array}
\usepackage{enumitem}
\usepackage{xcolor}
\usepackage{comment}
\usepackage[hidelinks,bookmarksnumbered=true]{hyperref}
\usepackage{algorithm,algpseudocode}
\usepackage[capitalize,sort]{cleveref}
\usepackage{tikz}
\usepackage{thm-restate}
\usepackage{multirow}
\usepackage[symbol]{footmisc}

\newtheorem{lemma}{Lemma}
\newtheorem{theorem}{Theorem}
\newtheorem{observation}{Observation}
\newtheorem{corollary}{Corollary}
\newtheorem{definition}{Definition}
\newtheorem{claim}{Claim}
\newtheorem*{claim*}{Claim}

\newtheorem{proposition}{Proposition}
\newtheorem{remark}{Remark}
\newenvironment{claimproof}[1]{\emph{Proof. }\space#1}{\hfill $\blacksquare$\medskip\par}

\newcommand{\lbdelta}{3/956}
\newcommand{\lbdeltafrac}{\frac{3}{956}}
\newcommand{\lbepsilon}{3/3836}
\newcommand{\lbepsilonfrac}{\frac{3}{3836}}
\newcommand{\lbnone}{n(\frac{1}{4}-\delta)/(\frac{1}{4}+\frac{\delta}{3})}
\newcommand{\mms}{\text{MMS}}
\newcommand{\MMS}{\mathrm{MMS}}
\newcommand{\I}{\mathcal{I}}

\title{Breaking the $3/4$ Barrier for Approximate Maximin Share}
\author{Hannaneh Akrami\thanks{Max Planck Institute for Informatics and Graduiertenschule Informatik, Universit\"at des Saarlandes}\\ \texttt{\small hakrami@mpi-inf.mpg.de} \and Jugal Garg\thanks{University of Illinois at Urbana-Champaign. Supported by NSF Grant CCF-1942321}\\ \texttt{\small jugal@illinois.edu}}
\date{}

\begin{document}
\maketitle
\begin{abstract}
    We study the fundamental problem of fairly allocating a set of indivisible goods among $n$ agents with additive valuations using the desirable fairness notion of maximin share (MMS). MMS is the most popular share-based notion, in which an agent finds an allocation fair to her if she receives goods worth at least her MMS value. An allocation is called MMS if all agents receive at least their MMS value. 
    Since MMS allocations need not exist when $n>2$, a series of works showed the existence of approximate MMS allocations with the current best factor of $\frac34 + O(\frac1n)$. However, a simple example in~\cite{deuermeyer1982scheduling, babaioff2021fair, simple} showed the limitations of existing approaches and proved that they cannot improve this factor to $3/4 + \Omega(1)$.
    In this paper, we bypass these barriers to show the existence of $(\frac{3}{4} + \lbepsilonfrac)$-MMS allocations by developing new reduction rules and analysis techniques.
\end{abstract}
\section{Introduction}
Fair allocation of resources (goods) is a fundamental problem in the intersection of computer science, economics, and social choice theory. This age-old problem arises naturally in a wide range of real-life settings, which was formally introduced in the seminal work of Steinhaus in the 1940s~\cite{steinhaus1948problem}. 
Depending on what properties the goods have and what notion of fairness is considered, one can address a wide range of problems. Extensive work has been done for the case of \emph{divisible} goods, where goods can be fractionally allocated, e.g.,~\cite{Varian74,Foley67,AzizM16b,AzizM16}. 

More recently, fair division of indivisible goods has received significant attention due to their applications in various multi-agent settings. Formally, an instance of fair division of indivisible goods consists of a set $N = \{1,2, \ldots, n\}$ of agents, a set $M$ of $m$ indivisible goods, and valuation vector $\mathcal{V} = (v_1, \ldots, v_n)$ where $v_i: 2^M \rightarrow \mathbb{R}_{\geq 0}$ is the valuation function of agent $i$. The goal is to find an allocation $A = \langle A_1, A_2, \ldots, A_n  \rangle$,  in which agent $i$ gets $A_i$, and $A$ satisfies some fairness criteria.

Two main categories of fairness are envy-based notions and share-based notions. Roughly speaking, in envy-based notions, an agent finds an allocation fair by comparing her bundle with other agents' bundles. Under allocation $A$, if certain conditions are met for all agents (e.g., $v_i(A_i) \geq v_i(A_j)$ for all $i,j \in N$ in the case of envy-freeness), then $A$ is fair. Popular examples of envy-based notions are envy-freeness (EF) and its relaxations envy-freeness up to any good (EFX) \cite{EFXjournal}, and envy-freeness up to one good (EF1) \cite{EF1}. 

In share-based notions, an agent finds an allocation fair only through the value she obtains from her bundle (irrespective of what others receive). For each agent $i$, if the value $i$ receives is at least some threshold $t_i$, then the allocation is said to be fair. An example of a share-based notion is proportionality. An allocation $A$ is proportional if all agents receive their proportional share, i.e., $v_i(A_i) \geq v_i(M)/n$ for all agents $i \in N$. It is easy to see that proportionality is too strong to be satisfied in the discrete setting.\footnote{As a counter-example, consider two agents and one good with a positive utility to both of the agents. Note that no matter how we allocate this good, one agent receives $0$ utility, which rules out the existence of proportional allocations and any approximation of proportionality.} This necessitates studying relaxed fairness notions when goods are indivisible. 

In this paper, we consider a natural relaxation of proportionality called \emph{maximin share (MMS)}, introduced by Budish \cite{budish2011combinatorial}. It is also preferred by participating agents over other notions, as shown in real-life experiments by~\cite{GatesGD20}. Maximin share of an agent is the maximum value she can guarantee to obtain if she divides the goods into $n$ bundles (one for each agent) and receives a bundle with the minimum value. Basically, for an agent $i$, assuming that all agents have $i$'s valuation function, the maximum value one can guarantee for all the agents is the $i$'s maximin share, denoted by $\mms_i$. Formally, for a set $S$ of goods and any positive integer $d$, let $\Pi_d(S)$ denote the set of all partitions of $S$ into $d$ bundles. Then,
\begin{align}
    \MMS_i^d(S) := \max_{P \in \Pi_d(S)} \min_{j=1}^d v_i(P_j).\nonumber 
\end{align}
For all agents $i$, $\MMS_i = \MMS_i^n(M)$. An allocation is MMS if all agents value their bundles at least as much as their MMS values. Formally, allocation $A$ is MMS if $v_i(A_i) \geq \mms_i$ for all agents $i \in N$. 

Since MMS allocations do not always exist when there are three or more agents with additive valuations \cite{procaccia2014fair,feige2021tight}, the focus shifted to study approximations of MMS. An allocation $A$ is $\alpha$-MMS if $v_i(A_i) \geq \alpha\cdot \mms_i$ for all agents $i \in N$. 
We note that the MMS notion is closely related to the popular max-min objective or the classic Santa Claus problem ($\max_A \min_i v_i(A_i)$)~\cite{BansalS06}. Unlike the max-min objective, the ($\alpha$-)MMS objective satisfies the desirable scale-invariance property. In the case of agents with identical valuations, an exact MMS allocation exists, and in this case, finding $\alpha$-MMS allocation is equivalent to $\alpha$-approximation of the Santa Claus problem. The best approximation factor known for the max-min objective under additive valuations is $\tilde{O}(m^\varepsilon)$ for any $\varepsilon>0$~\cite{ChakrabartyCK09}.

For the MMS problem, Procaccia and Wang \cite{procaccia2014fair} showed the existence of $2/3$-MMS allocations. Many follow-up works have improved the approximation factor \cite{barman2020approximation, ghodsi2018fair, garg2019approximating, amanatidis2017approximation, kurokawa2018fair, garg2021improved} with the current best result of $\alpha = \frac{3}{4} + \min(\frac{1}{36}, \frac{3}{16n-4})$ \cite{simple}. However, since the work of Ghodsi et al. \cite{ghodsi2018fair}, 
the best known constant approximation factor for MMS has remained $3/4$ for large $n$. In this work, we break this $3/4$ wall by proving the existence of $(\frac{3}{4} + \lbepsilonfrac)$-MMS allocations. 

After Ghodsi et al. \cite{ghodsi2018fair} proved the existence of $3/4$-MMS allocations and gave a PTAS to compute one, Garg and Taki \cite{garg2021improved} gave a simple algorithm with complicated analysis proving the existence of $(\frac{3}{4} + \frac{1}{12n})$-MMS allocations and also computing a $3/4$-MMS allocation in polynomial time. Very recently, Akrami et al. \cite{simple} simplified the analysis of (a slight modification of) the Garg-Taki algorithm significantly and proved the existence of $(\frac{3}{4} + \min(\frac{1}{36}, \frac{3}{16n-4}))$-MMS allocations. However, a simple example in~\cite{deuermeyer1982scheduling, babaioff2021fair, simple} shows that no constant factor better than $3/4$ can be obtained for approximate MMS using Garg-Taki algorithm. In Section \ref{technical}, we discuss the known techniques' barriers in more detail and how our algorithm overcomes these barriers.

\begin{table}[]
    \centering
    \begin{tabular}{|c|l|l|}
        \hline
        & \textbf{Existence} & \textbf{Non-existence} \\
        \hline
        \hline
        $n=3$ & $11/12$ \cite{feige2022improved}& $>39/40$ \cite{feige2021tight}\\
        \hline
        $n=4$ & $4/5$ \cite{ghodsi2018fair, babaioff2022fair}& $>67/68$ \cite{feige2021tight}\\
        \hline
        \multirow{5}{2.5em}{$n > 4$} & $2/3$  \cite{procaccia2014fair,amanatidis2017approximation,kurokawa2018fair,garg2019approximating}& \\ 
        & $2/3(1 + 1/(3n-1))$ \cite{barman2020approximation}& $> 1 - \mathcal{O}(\frac{1}{2^n})$ \cite{procaccia2014fair}\\
        & $3/4$ \cite{ghodsi2018fair}& \\
        & $3/4 + 1/(12n)$ \cite{garg2021improved}& $> 1 - \frac{1}{n^4}$ \cite{feige2021tight}\\
        & $3/4 + \min(1/36, 3/(16n-4))$ \cite{simple}& \\
        & \cellcolor{gray!25}$3/4 + \lbepsilon$ \textbf{(\cref{thm:main})}& \\
        \hline
    \end{tabular}
    \caption{Summary of the approximate MMS results when agents have additive valuations}
    \label{tab:my_label}
\end{table}

The complementary problem is to find upper bounds on the largest $\alpha$ for which $\alpha$-MMS allocations exist. Feige et al.~\cite{feige2021tight} constructed an example with three agents and nine goods for which no allocation is better than $39/40$-MMS. For $n \ge 4$, their construction gives an example for which no allocation is better than $(1-n^{-4})$-MMS.
Table~\ref{tab:my_label} summarizes all these results. We note that most of these existence results can be easily converted into PTAS for finding such an allocation using the PTAS for finding the MMS values~\cite{woeginger1997polynomial}. 

\subsection{Further related work}

\noindent{\bf Special cases.}
There has been a line of work on the instances with a limited number of agents or goods.
When $m \leq n+3$, an MMS allocation always exists \cite{amanatidis2017approximation}.
Feige et al. \cite{feige2021tight} improved this bound to $m \leq n+5$.
For $n=2$, MMS allocations always exist \cite{bouveret2016characterizing}.
For $n=3$, the MMS approximation was improved from $3/4$ \cite{procaccia2014fair}
to $7/8$ \cite{amanatidis2017approximation} to $8/9$ \cite{gourves2019maximin}, and then to $11/12$~\cite{feige2022improved}. For $n=4$, Ghodsi et al. \cite{ghodsi2018fair} showed the existence of $4/5$-MMS. For $n \geq 5$, the best known factor is the general $(\frac{3}{4} + \min(\frac{1}{36}, \frac{3}{16n-4}))$ bound given by Akrami et al. \cite{simple}.
\medskip

\noindent{\bf Ordinal approximation.} An alternative way of relaxing MMS is guaranteeing $1$-out-of-$d$ maximin share (MMS) for $d>n$, which is the maximum value that an agent can ensure by partitioning the goods into $d$ bundles and choosing the least preferred bundle. This notion only depends on the bundles' ordinal ranking and is not affected by a small perturbation in the value of every single good (as long as the ordinal ranking of the bundles does not change). A series of works studied this notion \cite{ordinalOne, hosseini2021guaranteeing} with the state-of-the-art being the existence of $1$-out-of-$\lfloor \frac{3n}{2} \rfloor$ MMS allocations for goods \cite{Hosseini2021OrdinalMS}. 
\medskip

\noindent{\bf Chores.} MMS can be analogously defined for fair division of chores. MMS allocations do not always exist for chores \cite{aziz2017algorithms}, which motivated the study of approximate MMS ~\cite{aziz2017algorithms,barman2020approximation}, with the current best approximation ratio being very recently improved from $11/9$ \cite{huang2021algorithmic} to $13/11$~\cite{huang2023reduction}. In the case of $n=3$, $19/18$-MMS allocations exist~\cite{feige2022improved}.

MMS in the chores setting is closely related to the well-studied variants of bin-packing and job scheduling problems. In particular, the recent paper~\cite{huang2023reduction} utilizes the Multifit algorithm for makespan minimization to obtain the best approximation factor. Therefore, many ideas which are already developed are proven to be useful when dealing with chores. 
On the other hand, when dealing with goods, the related variants of bin packing and scheduling problems do not make much sense where the objective becomes to maximize the number/capacity of bins or maximize the minimum processing time of a machine while allocating all the items. Therefore, new ideas specific to this problem are required. Furthermore, although the explicit study of MMS for goods started much before chores, the advancement in approximate MMS for chores has been faster. Also, the current best factor $(13/11)$ is much better than the analogous factor for goods $(3/4+\lbepsilon)$, despite the extensive work by many researchers on the goods problem.

For ordinal approximation, the best-known factor for existence is $1$-out-of-$\lfloor \frac{3n}{4} \rfloor$ MMS allocations for chores. The discrepancy carries on to the ordinal approximations of MMS. While the best known $d$ for which $1$-out-of-$d$ MMS allocations exist in the goods setting is $\lfloor 3n/2\rfloor$ \cite{Hosseini2021OrdinalMS}, the analogous factor for the chores setting is $\lfloor 3n/4\rfloor$ \cite{Hosseini2022OrdinalMS}.
\medskip

\noindent{\bf Other settings.} The MMS notion has also been studied when agents have more general valuations than additive, e.g.,~\cite{barman2020approximation,ghodsi2018fair,li2021fair,uziahu2023fair}. Generalizations have also been studied where restrictions are imposed on the set of feasible allocations, such as matroid constraints \cite{gourves2019maximin}, cardinality constraints \cite{biswas2018fair}, and graph connectivity constraints \cite{bei2022price,truszczynski2020maximin}. Strategyproof versions of fair division have also been studied~\cite{barman2019fair,amanatidis2016truthful,amanatidis2017truthful,aziz2019strategyproof}. MMS has also inspired other notions of fairness, like weighted MMS \cite{farhadi2019fair}, AnyPrice Share (APS) \cite{babaioff2021fair}, Groupwise MMS \cite{barman2018groupwise,chaudhury2021little}, $1$-out-of-$d$ share \cite{hosseini2021guaranteeing}, and self-maximizing shares \cite{babaioff2022fair}. MMS has also been studied in best-of-both-worlds settings, where both ex-ante and ex-post guarantees are sought~\cite{BabaioffEF22}.

\section{Preliminaries}
For all $n \in \mathbb{N}$, let $[n]=\{1,2, \ldots, n\}$. A fair division instance $\mathcal{I}=(N,M,\mathcal{V})$ consist of a set of agents $N=[n]$, a set of goods $M=[m]$ and a vector of valuation functions $\mathcal{V}=(v_1, v_2, \ldots, v_n)$ such that for all $i \in [n]$, $v_i:2^M \rightarrow \mathbb{R}_{\geq 0}$ indicates how much agent $i$ likes each subset of the goods. In this paper, we assume the valuation functions are additive, i.e., for all $i\in [n]$ and $S \subseteq M$, $v_i(S) = \sum_{g \in S} v_i(\{g\})$. For ease of notation, for all $g \in M$, we use $v_i(g)$ or $v_{i,g}$ instead of $v_i(\{g\})$.

For a set $S$ of goods and any positive integers $d$, let $\Pi_d(S)$ denote the set of all partitions of $S$ into $d$ bundles. Then for any  valuation function $v$,
\begin{align}
    \MMS_v^d(S) := \max_{P \in \Pi_d(S)} \min_{j=1}^d v(P_j). \label{MMS-def}    
\end{align}
When the instance $\mathcal{I}=(N,M,\mathcal{V})$ is clear from the context, we denote $\mms_{v_i}^n$ by $\mms_i(\I)$ or $\mms_i$ for all $i \in [n]$. For each agent $i$, let $P^i=(P^i_1, P^i_2, \ldots, P^i_n)$ be a partition of $M$ into $n$ bundles admitting the MMS value of agent $i$. Formally, $\mms_i = \min_{j \in [n]} v_i(P^i_j)$. We call such a partition, an MMS partition of agent $i$. An allocation $X$ is MMS if for all agents $i \in N$, $v_i(X_i) \geq \mms_i$. Similarly, for any $0 < \alpha \leq 1$, an allocation $X$ is $\alpha$-MMS if $v_i(X_i) \geq \alpha\cdot \mms_i$ for all agents $i \in N$. 
\begin{definition}[Ordered instance]
     An instance $\I=(N,M,\mathcal{V})$ is ordered if there exists an ordering of the goods $(g_1, g_2, \ldots, g_m)$ such that for all agents $i \in N$, $v_i(g_1) \geq v_i(g_2) \geq \ldots \geq v_i(g_m)$.
\end{definition}
It is known that the hardest instances of MMS are the ordered instances~\cite{barman2020approximation}. We use the notations used in \cite{simple}. 
\begin{definition}[\cite{simple}]
For the fair division instance $\I = ([n], [m], \mathcal{V})$,
$\mathtt{order}(\I)$ is defined as the instance $([n], [m], \mathcal{V}')$, where
for each $i \in [n]$ and $j \in [m]$, $v'_i(j)$ is
the $j^{\text{th}}$ largest number in the multiset $\{v_i(g) \mid g \in [m]\}$.
\end{definition}

The transformation $\mathtt{order}$ is $\alpha$-\emph{MMS-preserving}, i.e.,
for a fair division instance $\I$, given an $\alpha$-MMS allocation of $\mathtt{order}(\I)$, one can compute an $\alpha$-MMS allocation of $\I$ in polynomial time~\cite{barman2020approximation}. Given any ordered instance $\I=([n],[m],\mathcal{V})$, without loss of generality, we assume $v_i(1) \geq v_i(2) \geq \ldots \geq v_i(m)$ for all $i \in [n]$.
\begin{lemma}[\cite{barman2020approximation}]\label{order-preserves}
    Given an instance $\I$ and an $\alpha$-MMS allocation of $\mathtt{order}(\I)$, one can compute an $\alpha$-MMS allocation of $\I$ in polynomial time.
\end{lemma}
\begin{definition}[Normalized instance]
    An instance $\mathcal{I}=(N,M,\mathcal{V})$ is \emph{normalized}, if for all $i,j \in [n]$, $v_i(P^i_j)=1$.
\end{definition}
Note that since $v_i$ is additive, if $\I$ is normalized, then for all MMS partitions of $i$ like $Q=(Q_1, \ldots, Q_n)$ and for all $j \in [n]$ we have $v_i(Q_j)=1$.
\cite{simple} shows that given any instance $\mathcal{I}=(N,M,\mathcal{V})$, one can compute a normalized instance $\mathcal{I}'=(N,M,\mathcal{V'})$ such that any $\alpha$-MMS allocation for $\mathcal{I}'$ is an $\alpha$-MMS allocation for $\mathcal{I}$. Their algorithm converting an instance to a normalized instance is shown in \cref{algo:normalize}. We note that since finding an agent's MMS value is NP-hard, this is not a polynomial-time algorithm, but a PTAS exists.
\begin{algorithm}[!t]
\caption{$\mathtt{normalize}(N, M, \mathcal{V})$}
\label{algo:normalize}
\begin{algorithmic}[1]
\For{$i \in N$}
    \State Compute agent $i$'s MMS partition $P^i$.
    \State $\forall j \in N$, $\forall g \in P^i_j$, let $v'_{i,g} \leftarrow v_{i,g} / v_i(P^i_j)$.
\EndFor
\State \Return $(N, M, \mathcal{V}')$.
\end{algorithmic}
\end{algorithm}
\begin{lemma}[\cite{simple}]\label{thm:normalize}
Let $\I' = (N, M, \mathcal{V}') = \mathtt{normalize}(\I = (N, M, \mathcal{V}))$. Then for any allocation $A$,
$v_i(A_i) \ge v'_i(A_i)\MMS_i(\I)$ for all $i \in N$.
\end{lemma}
\Cref{thm:normalize} implies that $\mathtt{normalize}$ is $\alpha$-MMS-preserving,
since if $A$ is an $\alpha$-MMS allocation for the normalized instance $(N, M, \mathcal{V}')$,
then $A$ is also an $\alpha$-MMS allocation for the original instance $(N, M, \mathcal{V})$.
\cite{simple} give some structural property of ordered normalized instances which we repeat here in \cref{upper-13}. For completeness, we repeat its proof in Appendix~\ref{app:mp}.
\begin{restatable}{lemma}{upper}\cite{simple}\label{upper-13}
    Let $([n], [m], \mathcal{V})$ be an ordered and normalized fair division instance.
    For all $k \in [n]$ and agent $i \in [n]$, if $v_i(k) + v_i(2n-k+1) > 1$,
    then $v_i(2n-k+1) \leq 1/3$ and $v_i(k) > 2/3$.
\end{restatable}

\subsection{Reduction rules}
Given any instance $\mathcal{I}$, a reduction rule $R(\mathcal{I})$ is a procedure that allocates a subset $S \subseteq M$ of goods to an agent $i$ and outputs the instance $\mathcal{I'} = (N \setminus \{i\}, M \setminus S, \mathcal{V})$. 
\begin{definition}[Valid reductions]
    Let $R$ be a reduction rule and $R(\mathcal{I}) = (N',M',\mathcal{V})$ such that $\{i\} = N \setminus N'$ and $S = M \setminus M'$. Then $R$ is a ``valid $\alpha$-reduction'' if 
    \begin{enumerate}
        \item $v_i(S) \geq \alpha\cdot \mms^{|N|}_{v_i}(M)$, and
        \item for all $j \in N'$, $\mms^{|N|-1}_{v_j}(M') \geq \mms^{|N|}_{v_j}(M)$.
    \end{enumerate}

    Furthermore, a reduction rule $R$ is a ``valid reduction for agent $j \in N'$'', if $\mms^{|N|-1}_{v_j}(M') \geq \mms^{|N|}_{v_j}(M)$ where $N'$ and $M'$ are the set of remaining agents and remaining goods respectively after the reduction. 
\end{definition}
Note that if $R$ is a valid $\alpha$-reduction and an $\alpha$-MMS allocation $A$ exists for $R(\mathcal{I})$, then an $\alpha$-MMS allocation exists for $\mathcal{I}$. Such an allocation can be obtained by allocating $S$ to $i$ and allocating the rest of the goods as they are allocated under $A$.
\begin{lemma}\label{simple-valid}
    Given an instance $\I = (N, M, \mathcal{V})$, let $S \subseteq M$ be such that $v_i(S) \leq \mms_i$ and $|S| \leq 2$. Then allocating $S$ to an arbitrary agent $j \neq i$, is a valid reduction for agent $i$.
\end{lemma}
\begin{proof}
    Let $P = (P_1, P_2, \ldots, P_n)$ be an MMS partition of $M$ for agent $i$. Let $g_1, g_2 \in S$. In case $|S|=1$, $g_1=g_2$. Without loss of generality, we assume $g_1 \in P_1$. If $g_2 \in P_1$, then $(P_2, \ldots, P_n)$ is a partition of a subset of $M \setminus S$ into $n-1$ bundles with minimum value at least $\mms^n_{v_i}(M)$. Therefore, $\mms^{n-1}_{v_i}(M \setminus S) \geq \mms^n_{v_i}(M)$. In case $g_2 \notin P_1$, without loss of generality, let us assume $g_2 \in P_2$. Then $v_i(P_1 \cup P_2 \setminus S) = v_i(P_1) + v_i(P_2) - v_i(S) \geq \mms^n_{v_i}$. Therefore, $(P_1 \cup P_2 \setminus S, P_3, \ldots, P_n)$ is a partition of $M \setminus S$ into $n-1$ bundles with minimum value at least $\mms^n_{v_i}(M)$. Hence also in this case, $\mms^{n-1}_{v_i}(M \setminus S) \geq \mms^n_{v_i}(M)$. Thus, allocating $S$ to an arbitrary agent $j \neq i$, is a valid reduction for agent $i$.
\end{proof}

Now we define four reduction rules that we use in our algorithm. 
\begin{definition} For an ordered instance $\mathcal{I}=(N,M,\mathcal{V})$ and $\alpha > 0$, reduction rules $R_1^{\alpha}$, $R_2^{\alpha}$, $R_3^{\alpha}$ and $R_4^{\alpha}$ are defined as follows.
    \begin{itemize}
        \item \textbf{$R_1^{\alpha}(\mathcal{I}):$} If $v_i(1) \geq \alpha$ for some $i \in N$, allocate $\{1\}$ to agent $i$ and remove $i$ from $N$.
        \item \textbf{$R_2^{\alpha}(\mathcal{I}):$} If $v_i(\{2n-1, 2n, 2n+1\}) \geq \alpha$ for some $i \in N$, allocate $\{2n-1, 2n, 2n+1\}$ to agent $i$ and remove $i$ from $N$. 
        \item \textbf{$R_3^{\alpha}(\mathcal{I}):$} If $v_i(\{3n-2,3n-1,3n,3n+1\}) \geq \alpha$ for some $i \in N$, allocate $\{3n-2, 3n-1, 3n, 3n+1\}$ to agent $i$ and remove $i$ from $N$.
        \item \textbf{$R_4^{\alpha}(\mathcal{I}):$} If $v_i(\{1,2n+1\}) \geq \alpha$ for some $i \in N$, allocate $\{1,2n+1\}$ to agent $i$ and remove $i$ from $N$.
    \end{itemize}
\end{definition} 

We note that $R_1^{\alpha}$, $R_2^{\alpha}$, $R_4^{\alpha}$ in addition to one more rule of allocating $\{n,n+1\}$ to an agent is used in~\cite{garg2021improved, simple}. Our algorithm does not use the rule of allocating $\{n,n+1\}$. Moreover, $R_3^{\alpha}$ (allocating $\{3n-2,3n-1,3n,3n+1\}$) is used in our work and not elsewhere. 
\begin{lemma}\label{valid-rules}
    Given any $\alpha > 0$ and an ordered instance $\I$, $R_1^\alpha$, $R_2^\alpha$, and $R_3^\alpha$ are valid reductions for all the remaining agents.
\end{lemma}
\begin{proof}
    For a remaining agent $i$, let $P=(P_1, \ldots, P_n)$ be an MMS partition of $M$ for $i$. It suffices to prove that after each of these reduction rules, there exists a partition of the remaining goods for each remaining agent into $n-1$ bundles with a minimum value of $\mms^n_i(M)$ for agent $i$.
    \begin{itemize}
        \item $R^\alpha_1$: Let $1 \in P_k$. Then removing $P_k$ from $P$ results in a partition of a subset of $M \setminus \{1\}$ into $n-1$ bundles of value at least $\mms^n_i(M)$ for agent $i$.
        \item $R^\alpha_2$: By the pigeonhole principle, there exists $k$ such that $|P_k \cap \{1,2, \ldots, 2n+1\}| \geq 3$. Let $g_1, g_2, g_3 \in P_k \cap \{1,2, \ldots, 2n+1\}$ and $g_1 < g_2 < g_3$. Replace $g_1$ with $2n-1$, $g_2$ with $2n$ and $g_3$ with $2n+1$ and remove $P_k$ from $P$. Note that the value of the remaining bundles can only increase. Thus, the result is a partition of a subset of $M \setminus \{2n-1,2n,2n+1\}$ into $n-1$ bundles with a minimum $\mms^n_i(M)$ for agent $i$.
        \item $R^\alpha_3$: The proof is very similar to $R^\alpha_2$ case. 
	   By the pigeonhole principle, there exists $k$ such that $|P_k \cap \{1,2, \ldots, 3n+1\}| \geq 4$. Let $g_1, g_2, g_3, g_4 \in P_k \cap \{1,2, \ldots, 3n+1\}$ and $g_1 < g_2 < g_3 < g_4$. Replace $g_1$ with $3n-2$, $g_2$ with $3n-1$, $g_3$ with $3n$ and $g_4$ with $3n+1$ and remove $P_k$ from $P$. Note that the value of the remaining bundles can only increase. Thus, the result is a partition of a subset of $M \setminus \{3n-2,3n-1,3n,3n+1\}$ into $n-1$ bundles with a minimum value of $\mms^n_i(M)$ for agent $i$. \qedhere
    \end{itemize}
\end{proof}
\begin{proposition}\label{thm:vr-upper-bounds}
    If $\mathcal{I}$ is ordered and for a given $\alpha \geq 0$, none of the rules $R^\alpha_1$, $R^\alpha_2$ or $R^\alpha_3$ is applicable, then
    \begin{enumerate}
        \item for all $k \geq 1$, $v_i(k) < \alpha$, and
        \item for all $k > 2n$, $v_i(k) < \alpha/3$, and
        \item for all $k > 3n$, $v_i(k) < \alpha/4$.
    \end{enumerate}
\end{proposition}
\begin{proof}
    We prove each case separately.
    \begin{enumerate}
        \item Since $R^\alpha_1$ is not applicable, $v_i(k) \leq v_i(1) < \alpha$ for all agents $i$ and all $k \geq 1$.
        \item Since $R^\alpha_2$ is not applicable, $3v_i(k) \leq 3v_i(2n+1) \leq v_i(2n-1)+v_i(2n)+v_i(2n+1) < \alpha$ for all agents $i$ and all $k>2n$. Therefore, $v_i(k) < \alpha/3$.
        \item Similar to the former case, since $R^\alpha_3$ is not applicable, $4v_i(k) \leq 4v_i(3n+1) \leq v_i(3n-2) + v_i(3n-1)+v_i(3n)+v_i(3n+1) < \alpha$ for all agents $i$ and all $k>3n$. Therefore, $v_i(k) < \alpha/4$. \qedhere
    \end{enumerate} 
\end{proof}

\begin{definition}[$\alpha$-irreducible and $\delta$-ONI]
We call an instance $\mathcal{I}$ $\alpha$-irreducible if none of the rules $R^\alpha_1$, $R^\alpha_2$, $R^\alpha_3$ or $R^\alpha_4$ is applicable. Moreover, we call an instance $\delta$-ONI if it is ordered, normalized, and $(3/4 + \delta)$-irreducible. 
\end{definition}

\section{Technical overview}\label{technical}
Most algorithms for approximating MMS, especially those with a factor of at least $3/4$~\cite{ghodsi2018fair,garg2021improved,simple}, utilize two simple tools: valid reductions and bag filling. Although these tools are easy to use in a candidate algorithm, the novelty of these works is in the analysis, which is challenging. Like previous works, the analysis is the most difficult part of our algorithm based on these tools. Unlike previous works, we also need to use a new reduction rule and initialize bags differently, which are counterintuitive.

First, we discuss the algorithm given by \cite{simple}, which is a slight modification of the algorithm in \cite{garg2021improved}. For $\alpha \leq 3/4$, \cite{simple} showed how to obtain an ordered normalized $\alpha$-irreducible instance from any arbitrary instance such that the transformation is $\alpha$-MMS preserving.\footnote{\cite{simple} uses $R^\alpha_1$, $R^\alpha_2$, $R^\alpha_4$ and one more rule as reduction rules. However, all that matters in their proof is that the applied reduction rules are valid $\alpha$-reduction rules.} That is, given an $\alpha$-MMS allocation for the resulting ordered normalized irreducible instance, one can obtain an $\alpha$-MMS allocation for the original instance. In the first phase of their algorithm, they obtain an ordered normalized $\alpha$-irreducible instance $\hat{\I}$ and in the second phase, they compute an $\alpha$-MMS allocation for $\hat{\I}$. Let $\hat{\I} = ([n], [m], \mathcal{V})$. Without loss of generality, we can assume that $m \geq 2n$ (\cref{2n-goods}).

In the second phase, they initialize $n$ bags with the first $2n$ goods as follows. 
\begin{equation}
    \label{eq:B_i}
    B_k := \{k , 2n-k+1\} \text{ for } k\in [n]
\end{equation}
See Figure \ref{B-bags} for a better intuition.
\begin{figure}[t]\centering
\begin{tikzpicture}
[scale=1,
 good/.style={circle, draw=black, thick, minimum size=30pt},
]

\draw[black, very thick] (-0.4-0.25,0.8) rectangle (0.4+0.25,3.4);

\node[good]      at (0,2.75)      {$\scriptstyle{2n}$};
\node[good]      at (0,1.5)      {$\scriptstyle{1}$};

\node at (0, 0.5) {$\scriptstyle{B_1}$};

\filldraw[color=black!60, fill=black!5, thick](1.6, 2) circle (0.02);
\filldraw[color=black!60, fill=black!5, thick](1.7, 2) circle (0.02);
\filldraw[color=black!60, fill=black!5, thick](1.8, 2) circle (0.02);

\draw[black, very thick] (-1.4-0.25 +4.5,0.8) rectangle (-0.6+0.25 +4.5,3.4);

\node[good, scale=0.74]      at (0 +3.5,2.75)      {$\scriptstyle{2n-k+1}$};
\node[good]      at (0 +3.5,1.5)      {$\scriptstyle{k}$};

\node at (0 +3.5, 0.5) {$\scriptstyle{B_k}$};

\filldraw[color=black!60, fill=black!5, thick](0.1+5, 2) circle (0.02);
\filldraw[color=black!60, fill=black!5, thick](0.2+5, 2) circle (0.02);
\filldraw[color=black!60, fill=black!5, thick](0.3+5, 2) circle (0.02);

\draw[black, very thick] (-0.4-5.25 +12,0.8) rectangle (0.4+0.25-5 +12,3.4);

\node[good]      at (0 +7,2.75)      {$\scriptstyle{n+1}$};
\node[good]      at (0 +7,1.5)      {$\scriptstyle{n}$};

\node at (0 +7, 0.5) {$\scriptstyle{B_n}$};

\end{tikzpicture}
\caption{Configuration of Bags $B_1, B_2, \ldots, B_n$}
\label{B-bags}
\end{figure}
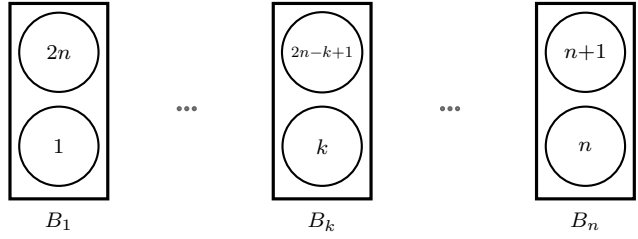
As long as an agent $i$ values a bag $B_k$ at least $\alpha$, allocate $B_k$ to $i$ and remove $B_k$ and $i$. Then, as long as an unallocated bag exists (and thus a remaining agent), pick an arbitrary remaining bag $B_k$ and add unassigned goods $g>2n$ until some remaining agent $i$ values it at least $\alpha$. Then, allocate $B_k$ to $i$ and continue. The second phase is called the bag-filling phase. Algorithm \ref{algo:bagFill} shows the pseudocode of the bag-filling phase of \cite{simple}.

\begin{algorithm}[!t]
\caption{$\mathtt{bagFill}(\I, \alpha)$}
\label{algo:bagFill}
\textbf{Input:} Ordered normalized $\alpha$-irreducible instance $\I = ([n], [m], \mathcal{V})$ and approximation factor $\alpha$ \\
\textbf{Output:} (Partial) allocation $A = \langle A_1, \ldots, A_n \rangle$.
\begin{algorithmic}[1]
\For{$k \in [n]$}
    \State $B_k = \{k, 2n+1-k\}$.
\EndFor
\State $U_G = [m] \setminus [2n]$  \Comment{unassigned goods}
\State $U_A = [n]$  \Comment{unsatisfied agents}
\State $U_B = [n]$  \Comment{unassigned bags}
\While{$U_A \neq \emptyset$}
    \If{$\exists i \in U_A$, $\exists k \in U_B$, such that $v_i(B_k) \ge \alpha$}
        \State $A_i = B_k$
        \State $U_A = U_A \setminus \{i\}$
        \State $U_B = U_B \setminus \{k\}$
    \Else
        \State $g$ = arbitrary good in $U_G$
        \State $k$ = arbitrary bag in $U_B$
        \State $B_k = B_k \cup \{g\}$.
        \State $U_G = U_G \setminus \{g\}$
    \EndIf
\EndWhile
\State \Return $\langle A_1, \ldots, A_n \rangle$
\end{algorithmic}
\end{algorithm}

To prove that the algorithm's output is $\alpha$-MMS, it suffices to prove that we never run out of goods in the bag-filling phase or, equivalently, all agents receive a bag at some point during the algorithm. To prove this, they categorize agents into two groups. Let $N^1 = \{i \in N \mid \forall k \in [n]: v_i(B_k) \leq 1 \}$ and $N^2 = N \setminus N^1 = \{i \in N \mid \exists k \in [n]: v_i(B_k) > 1 \}$. We note that the sets $N^1$ and $N^2$ are defined based on the instance $\hat{\mathcal{I}}$ at the beginning of phase $2$, and they do not change throughout the algorithm.

\paragraph{\boldmath Agents in $N^1$} Proving that all agents in $N^1$ receive a bag is easy. Using the fact that at the beginning of Phase $2$, the instance is ordered, normalized, and $\alpha$-irreducible, they prove $v_i(g) < 1/4$ for all $i \in N$ and all $g \in M \setminus [2n]$. This helps to prove that any bag which is not assigned to an agent $i \in N^1$ while $i$ was available has a value at most $1$ to $i$. Therefore, since $v_i(M)=n$, running out of goods is impossible before agent $i$ receives a bag. 

\paragraph{\boldmath Agents in $N^2$} The main bulk and difficulty of the analysis of \cite{garg2021improved} is to prove that all agents in $N^2$ receive a bag. By normalizing the instance, \cite{simple} managed to simplify this argument significantly. \cite{simple} prove $v_i(g) < 1/12$ for all $i \in N^2$ and all $g \in M \setminus [2n]$. This helps to bound the value of the bags that receive some goods in the bag-filling phase by $5/6$ for all available $i \in N^2$. Again, if the number of such bags is high enough, it is easy to prove that the algorithm does not run out of goods in the bag-filling phase. The difficult case is when the total value of the bags which are of value more than $1$ to some agent $i \in N^2$ is large. Roughly speaking, in this case, it seems that the bags which receive goods in the bag-filling phase and their values are bounded by $5/6$ cannot compensate for the large value of the bags that do not require any goods in the bag-filling phase. This is where the normalized property of $\hat{\I}$ simplifies the matter significantly. Intuitively, there are many goods with a high value that happened to be paired in the same bag in the bag initialization phase. Since the instance is normalized, we know that in the MMS partition of $i$, these goods cannot be in the same bag. This implies that many bags in the MMS partition of $i$ have at most $1$ good in common with the goods in $[2n]$. This means that the value of the remaining goods (the goods in $M \setminus [2n]$) must be large since they fill the bags in the MMS partition such that the value of each bag equals $1$. Hence, enough goods remain in $M \setminus [2n]$ to fill the bags. 

There are two main obstacles to generalizing this algorithm to obtain $\alpha$-MMS allocations when $\alpha>3/4$. The first obstacle lies in the first phase of the algorithm. $R^\alpha_4$ is a valid $\alpha$-reduction when $\alpha \leq 3/4$ and $R^\alpha_1$ and $R^\alpha_2$ are not applicable. This no longer holds when $\alpha > 3/4$. In this case, the MMS value of the agents can indeed decrease after applying $R^\alpha_4$. When $\alpha = 3/4 + \mathcal{O}(1/n)$, \cite{garg2021improved} and \cite{simple} managed to resolve this issue by adding some dummy goods after each iteration of $R^\alpha_4$ and proving that the total value of these dummy goods is negligible. Essentially, since we only need to guarantee the last agent a value of $\alpha$, the idea is to divide the excess $1-\alpha$ among all agents and improve the factor. However, this can only improve the factor by at most $\mathcal{O}(1/n)$. If $\alpha>3/4+\epsilon$ for a constant $\epsilon>0$, the same technique does not work since the value of dummy goods cannot be reasonably bounded. We resolve this issue in Section \ref{hard-sec}. Unlike the previous works, we allow the MMS values of the remaining agents to drop. Although the MMS values of the agents can drop, we show that they do not drop by more than a multiplicative factor of $(1-4\epsilon)$ after an arbitrary number of applications of $R^{3/4+\epsilon}_k$ for $k \in [4]$. Basically, while for $\alpha \leq 3/4$, one can get $\alpha$-irreduciblity for free (i.e., without losing any approximation factor on MMS), for $\alpha = 3/4+\epsilon$ and $\epsilon>0$, we lose an approximation factor of $(1-4\epsilon)$.

The second obstacle is that for goods in $M \setminus [2n]$, we do not get the neat bound of $v_i(g) < 1/4$ for $i \in N$. Instead, we get this bound with an additive factor of $\mathcal{O(\epsilon)}$. This even complicates the analysis for agents in $N^1$, which was trivial in~\cite{simple}. Furthermore, a tight example in~\cite{deuermeyer1982scheduling, babaioff2021fair, simple} shows that this algorithm cannot do better than $3/4+\mathcal{O}(1/n)$ and all the agents are in $N^1$ in this example. To overcome this hurdle, we further categorize the agents in $N^1$. One group consists of the agents with a reasonable bound on the value of good $2n+1$, and the other agents, the \emph{problematic} ones, do not. 

We break the problem into two cases depending on the number of these problematic agents. In Section \ref{small-N11}, we consider the case when the number of problematic agents is not too large. In this case, we work with a slight modification of the algorithm in \cite{simple}, and using an involved analysis, we show that it gives a $(3/4+\epsilon)$-MMS allocation. Otherwise, we introduce a new reduction rule for the first time that allocates the two most valuable goods to an agent. Although allocating these goods seem counterintuitive, surprisingly, that seems to be the only way to obtain a $(3/4+\epsilon)$-MMS allocation for the tight example in \cite{deuermeyer1982scheduling,babaioff2021fair, simple}. In Section \ref{large-N11}, we give another algorithm to handle the case where the number of problematic agents is too large. In this case, we first apply the reduction rules (including the new one), and then initialize the bags with three goods, unlike the previous works. Precisely, we set $C_k := \{k, 2n-k+1, 2n+k\}$ and then do bag-filling.

To summarize, the structure of the rest of the paper is as follows. In Section \ref{hard-sec}, given any instance $\I=(N,M,\mathcal{V})$ and $\epsilon > 0$, for $\delta \geq 4\epsilon/(1-4\epsilon)$ we obtain an ordered normalized $(3/4+\delta)$-irreducible ($\delta$-ONI) instance $\I'=(N',M',\mathcal{V}')$ such that $N' \subseteq N$, $M' \subseteq M$ and all agents in $N \setminus N'$ receive a bag of value at least $(3/4 +\epsilon) \mms_i(\I)$. Moreover, we prove from any $(3/4+\delta)$-MMS allocation for $\I'$, one can obtain a $\min\left(3/4+\epsilon, (3/4 + \delta) (1 - 4\epsilon) \right)$-MMS allocation for $\I$. 

In Section \ref{easy-sec}, we prove a $(3/4+\delta)$-\mms~allocation exists for all $\delta$-ONI instances for any $\delta \leq \lbdelta$.  Therefore, we prove that for $4\epsilon/(1-4\epsilon) \leq \delta \leq \lbdelta$, a $\min\left(3/4+\epsilon, (3/4 + \delta) (1 - 4\epsilon) \right)$-MMS exists for all instances. Setting $\delta = \lbdelta$ and $\epsilon = \delta/(4(\delta+1)) = \lbepsilon$, there always exists a $(3/4+\lbepsilon)$-\mms~allocation.

\section{\boldmath Reduction to $\delta$-ONI instances}\label{hard-sec}
In this section, for any $\epsilon>0$ and $\delta \geq 4\epsilon/(1-4\epsilon)$ we show how to obtain a $\delta$-ONI instance $\I'$ from any arbitrary instance $\I$, such that from any $\alpha$-MMS allocation for $\I'$, one can obtain a $\min\left(3/4+\epsilon, (1-4\epsilon) \alpha \right)$-MMS allocation for $\I$. To obtain such an allocation, first, we obtain a $(3/4+\epsilon)$-irreducible instance, and we prove that the MMS value of no remaining agent drops by more than a multiplicative factor of $(1-4\epsilon)$. Then, we normalize and order the resulting instance, giving us a $\delta$-ONI instance (for $\delta \geq 4\epsilon/(1-4\epsilon)$). In the rest of this section, by $R_k$ we mean $R^{(3/4+\epsilon)}_k$ for $k \in [4]$.

We start with transforming the instance into an ordered one using the $\mathtt{order}$ algorithm. Then we scale the valuations such that for all $i \in N$, $\mms_i = 1$. Then, as long as one of the reduction rules $R_1$, $R_2$, $R_3$, or $R_4$ is applicable, we apply $R_k$ for the smallest possible $k$. \cref{algo:reduce} shows the pseudocode of this procedure.

\begin{algorithm}[!t]
\caption{$\mathtt{reduce}((N, M, \mathcal{V}), \epsilon)$}
\label{algo:reduce}
\begin{algorithmic}[1]
\State $\I \leftarrow \mathtt{order}(N,M,\mathcal{V})$
\For{$i \in N$}
    \State $v_{i,g} \leftarrow v_{i,g}/\mms_i, \forall g\in [m]$
\EndFor
\While{$R^{(3/4+\epsilon)}_1$ or $R^{(3/4+\epsilon)}_2$ or $R^{(3/4+\epsilon)}_3$ or $R^{(3/4+\epsilon)}_4$ is applicable}
    \State $\I \leftarrow R^{(3/4+\epsilon)}_k(\I)$ for smallest possible $k$
\EndWhile
\State \Return $\I$.
\end{algorithmic}
\end{algorithm}
In this section, we prove the following two theorems. 
\begin{restatable}{theorem}{theoremOne}\label{thm:1}
    Given an instance $\I = (N, M, \mathcal{V})$ and $\epsilon \geq 0$, let $\I' = (N', M', \mathcal{V}') = \mathtt{reduce}(\I, \epsilon)$. For all agents $i \in N'$, $\mms_i(\I') \geq 1-4\epsilon$.
\end{restatable}
\begin{restatable}{theorem}{theoremTwo}\label{thm:2}
    Given an instance $\I$ and $\epsilon \geq 0$, let $\hat{\I} = \mathtt{order}(\mathtt{normalize}(\mathtt{reduce}(\I, \epsilon)))$. Then $\hat{\I}$ is ordered, normalized and $(\frac{3}{4}+\frac{4\epsilon}{1-4\epsilon})$-irreducible ($\frac{4\epsilon}{1-4\epsilon}$-ONI). Furthermore, from any $\alpha$-MMS allocation of $\hat{\I}$ one can obtain a $\min(3/4+\epsilon, (1-4\epsilon)\alpha)$-MMS allocation of $\I$.  
\end{restatable}

Note that once $R_1$ is not applicable, we have $v_i(1) < 3/4 + \epsilon$ for all remaining agents $i$. Since we never increase the values, $R_1$ can no longer apply. So $\mathtt{reduce}(\I, \epsilon)$ first applies $R_1$ as long as it is applicable and then applies the rest of the reduction rules.
Since $R_1$ is a valid reduction rule for all the remaining agents $i$ by \cref{valid-rules}, $\mms_i \geq 1$ after applications of $R_1$. So to prove \cref{thm:1} without loss of generality, we assume $R_1$ is not applicable on $\I = ([n], M, \mathcal{V})$. Let $\I' = (N',M', \mathcal{V}) = \mathtt{reduce}(\I, \epsilon)$. For the rest of this section, we fix agent $i = N'$. Let $P = (P_1, P_2, \ldots, P_n)$ be the initial MMS partition of $i$ (in $\I$). We construct a partition $Q = (Q_1, Q_2, \ldots, Q_{|N'|})$ of $M'$ such that $v_i(Q_j) \geq 1-4\epsilon$ for all $j \in [|N'|]$. 

Let $G_2$, $G_3$, and $G_4$ be the set of goods removed by applications of $R_2$, $R_3$, and $R_4$, respectively. Also, let $r_2 = |G_2|/3$, $r_3 = |G_3|/4$, and $r_4 = |G_4|/2$ be the number of times each rule is applied, respectively. Note that in the end, all that matters is that we construct a partition $Q$ of $M \setminus (G_2 \cup G_3 \cup G_4)$ into $n-(r_2+r_3+r_4)$ bundles of value at least $1-4\epsilon$ for $i$. For this sake, it does not matter in which order the goods are removed. Therefore, without loss of generality, we assume all the goods in $G_4$ are removed first, and then the goods in $G_2$ and $G_3$ are removed in their original order. Note that we are not applying the reduction rules in a different order. We are removing the same goods that would be removed by applying the reduction rules in their original order. Only for the sake of our analysis, we remove these goods in a different order. For better intuition, consider the following example. Assume $\mathtt{reduce}(\I, \epsilon)$ first applies $R_2$ removing $\{a_1, a_2, a_3\}$, then $R_4$ removing $\{b_1, b_2\}$, then another $R_2$ removing $\{c_1, c_2, c_3\}$ and then $R_3$ removing $\{d_1, d_2, d_3, d_4\}$. Without loss of generality we can assume that first $\{b_1, b_2\}$ is removed, then $\{a_1, a_2, a_3\}$, then $\{c_1, c_2, c_3\}$ and then $\{d_1, d_2, d_3, d_4\}$.

We know that when there are $n$ agents, removing $\{2n-1,2n,2n+1\}$ (or $\{3n-2,3n-1,3n,3n+1\}$) and an agent is a valid reduction for $i$ by \cref{valid-rules}. With the same argument, it is not difficult to see that removing $\{g_1, g_2, g_3\}$ where $g_1 \geq 2n-1$, $g_2 \geq 2n$ and $g_3 \geq 2n+1$ (or $\{g_1,g_2,g_3,g_4\}$ where $g_1 \geq 3n-2$, $g_2 \geq 3n-1$, $g_3 \geq 3n$ and $g_4 \geq 3n+1$) and an agent is also a valid reduction for $i$. For completeness, we prove this in \cref{redund}.
\begin{lemma}\label{redund}
    Let $\I = (N, M, \mathcal{V})$ be an ordered instance and $i \in N$.
    \begin{enumerate}
        \item Let $g_1 \geq 2n-1$, $g_2 \geq 2n$ and $g_3 \geq 2n+1$. Then $\mms^{n-1}_{v_i}(M \setminus \{g_1,g_2,g_3\}) \geq \mms^n_{v_i}(M)$.
        \item Let $g_1 \geq 3n-2$, $g_2 \geq 3n-1$, $g_3 \geq 3n$ and $g_4 \geq 3n+1$. Then $\mms^{n-1}_{v_i}(M \setminus \{g_1,g_2,g_3,g_4\}) \geq \mms^n_{v_i}(M)$.
    \end{enumerate}
\end{lemma}
\begin{proof}
    \begin{enumerate}
        \item By the pigeonhole principle, there exists $k$ such that $|P_k \cap \{1,2, \ldots, 2n+1\}| \geq 3$. Let $h_1, h_2, h_3 \in P_k \cap \{1,2, \ldots, 2n+1\}$ and $h_1 < h_2 < h_3$. Replace $h_1$ with $g_1$, $h_2$ with $g_2$ and $h_3$ with $g_3$ and remove $P_k$ from $P$. Note that the value of the remaining bundles can only increase. Thus, the result is a partition of a subset of $M \setminus \{g_1,g_2,g_3\}$ into $n-1$ bundles with a minimum value of $\mms^n_i(M)$ for agent $i$.
        \item By the pigeonhole principle, there exists $k$ such that $|P_k \cap \{1,2, \ldots, 3n+1\}| \geq 4$. Let $h_1, h_2, h_3, h_4 \in P_k \cap \{1,2, \ldots, 3n+1\}$ and $h_1 < h_2 < h_3 < h_4$. Replace $h_1$ with $g_1$, $h_2$ with $g_2$, $h_3$ with $g_3$ and $h_4$ with $g_4$ and remove $P_k$ from $P$. Note that the value of the remaining bundles can only increase. Thus, the result is a partition of a subset of $M \setminus \{g_1,g_2,g_3,g_4\}$ into $n-1$ bundles with a minimum value of $\mms^n_i(M)$ for agent $i$. \qedhere
    \end{enumerate}
\end{proof}

\begin{observation}\label{two-removal}
    Given an ordered instance $\I = (N, M, \mathcal{V})$, let $v_i(g_1) \geq \ldots \geq v_i(g_m), \forall i\in N$. Let $\mathcal{I'} = (N', M', \mathcal{V})$ be the instance after removing an agent $i$ and a set of goods $\{a,b\}$ from $\I$. Let $g \in M'$ be the $j^{\text{th}}$ most valuable good in $M$ and the ${j'}^{\text{th}}$ most valuable good in $M'$. Then $j' \geq j-2$. 
\end{observation}
\begin{corollary}[of \cref{two-removal}]\label{two-removal-cor}
    Given an ordered instance $\I = (N, M, \mathcal{V})$, let $\mathcal{I'} = (N', M', \mathcal{V})$ be the instance after removing an agent $i$ and a set of goods $\{a,b\}$ from $\I$. Let $n = |N|$ and $n'=|N'|=n-1$. Let $g \in M'$ be the $j^{\text{th}}$ most valuable good in $M$ and the ${j'}^{\text{th}}$ most valuable good in $M'$. Then,
    \begin{itemize}
        \item for any $k$, in particular, $k \in \{-1,0,1\}$, if $j \geq 2n+k$, then $j' \geq 2n'+k$, and
        \item for any $k$, in particular, $k \in \{-2,-1,0,1\}$, if $j \geq 3n+k$, then $j' \geq 3n'+k$.
    \end{itemize}
\end{corollary}

Next, assume at a step where the number of agents is $n$, $\{g_{2n-1}, g_{2n}, g_{2n+1}\}$ (or $\{g_{3n-2}, g_{3n-1}$, $g_{3n}, g_{3n+1}\}$) is removed with an application of $R_2$ (or $R_3$). \cref{two-removal-cor} together with \cref{redund} imply that removing $\{g_{2n-1}, g_{2n}, g_{2n+1}\}$ (or $\{g_{3n-2}, g_{3n-1}, g_{3n}, g_{3n+1}\}$) at a later step where the number of agents is $n' \leq n$ is also valid for agent $i$. Therefore, all that remains is to prove that after removing the goods in $G_4$ and $r_4$ agents, the MMS value of $i$ remains at least $1-4\epsilon$. That is, $\mms^{n-r_4}_i(M \setminus G_4) \geq 1-4\epsilon$. 

\begin{lemma}\label{long-lemma}
    Let $(N', M', \mathcal{V}) = \mathtt{reduce}(([n], M, \mathcal{V}), \epsilon)$. Let $r_4$ be the number of times $R_4$ is applied during $\mathtt{reduce}(\I, \epsilon)$ and let $G_4$ be the set of removed goods by applications of $R_4$.
    Then for all agents $i \in N'$, $\mms^{n-r_4}_{v_i}(M \setminus G_4) \geq 1-4\epsilon$.
\end{lemma}
\begin{proof}
    Without loss of generality, assume all the goods in $G_4$ are in $P_1 \cup P_2 \cup \ldots \cup P_k$ for some $k\le 2r_4$. Namely, we have $P_j \cap G_4 \neq \emptyset$ for all $j \in [k]$ and $(P_{k+1} \cup \ldots \cup P_n) \cap G_4 = \emptyset$. If $k \leq r_4$, then $(P_{k+1}, \ldots, P_n)$ is already a partition of a subset of $M \setminus G_4$ into at least $n-r_4$ bundles. Therefore the lemma follows.

    So assume $k > r_4$. In each application of $R_4$, two goods $h$ and $\ell$ are removed. Let $h$ be the more valuable good. We call $h$ the heavy good and $\ell$ the light good of this application of $R_4$. By \cref{thm:vr-upper-bounds}, for all heavy goods $h$ and light goods $\ell$ we have, $v_i(h) < 3/4+\epsilon$ and $v_i(\ell) < 1/4+\epsilon/3$. 
    Let $H$ be the set of all heavy goods and $L$ be the set of all light goods removed during these reductions. Hence, $G_4 = H \cup L$, $|H| = |L| = r_4$.
    
    We prove that we can partition $(P_1 \cup \ldots P_k) \setminus G_4$ into $k-r_4$ bundles $Q_1, \ldots, Q_{k-r_4}$, each of value at least $1-4\epsilon$. Or equivalently $\mms^{k-r_4}_{v_i}\big((P_1 \cup \ldots \cup P_k) \setminus G_4\big) \geq 1-4\epsilon$. Then, $(Q_1, \ldots, Q_{k-r_4}, P_{k+1}, \ldots, P_n)$ is a partition of $M \setminus G_4$ into $n-r_4$ bundles, each of value at least $1-4\epsilon$ and the lemma follows. It suffices to prove the following claim.
    \begin{claim*}
        For $r < k$ $\leq 2r$, if $|(P_1 \cup P_2 \cup \ldots \cup P_k) \cap H| \leq r$ and $|(P_1 \cup P_2 \cup \ldots \cup P_k) \cap L| \leq r$, then $\mms^{k-r}_{v_i}\big((P_1 \cup \ldots \cup P_k) \setminus G_4\big) \geq 1-4\epsilon$ for all $0<r<k$. 
    \end{claim*}
    The proof of the claim is by induction on $k$. For $k=2$, we have $r = 1$ and $v_i(P_1 \cup P_2) - v_i(H \cup L) \geq 2 - (\frac{3}{4}+\epsilon)-(\frac{1}{4}+\frac{\epsilon}{3}) > 1-4\epsilon$ and therefore, $\mms^1_{v_i}(P_1 \cup P_2 \setminus G_4) \geq 1-4\epsilon$. 
    Now assume that the statement holds for all values of $k' \leq k-1$, and we prove it for $k>2$. First, we prove the claim when at least one of the inequalities is strict. Assume $|(P_1 \cup P_2 \cup \ldots \cup P_k) \cap H| < r$ and $|(P_1 \cup P_2 \cup \ldots \cup P_k) \cap L| \leq r$. The proof of the other case is symmetric. If $(P_1 \cup P_2 \cup \ldots \cup P_k) \cap L \neq \emptyset$, without loss of generality assume $P_k \cap L \neq \emptyset$. Therefore, $|(P_1 \cup \ldots \cup P_{k-1}) \cap H|$ $\leq r-1 < k-1$ and $|(P_1 \cup \ldots \cup P_{k-1}) \cap L|$ $\leq r-1 < k-1$. We have,
    \begin{align*}
        \mms^{k-r}_{v_i}\big((P_1 \cup \ldots \cup P_k) \setminus G_4\big) &\geq 
        \mms^{(k-1)-(r-1)}_{v_i}\big((P_1 \cup \ldots \cup P_{k-1}\big) \setminus G_4) \\
        &\geq 1-4\epsilon. \tag{by induction assumption}
    \end{align*}
    Now assume $|(P_1 \cup P_2 \cup \ldots \cup P_k) \cap H| = r$ and $|(P_1 \cup P_2 \cup \ldots \cup P_k) \cap L| = r$.
    \paragraph{\boldmath Case 1: There exists $j \in [k]$, such that $P_j \cap H \neq \emptyset$ and $P_j \cap L \neq \emptyset$.}
    Without loss of generality, assume $P_k \cap H \neq \emptyset$ and $P_k \cap L \neq \emptyset$. In this case $|(P_1 \cup \ldots \cup P_{k-1}) \cap H|$ $\leq r-1 < k-1$ and $|(P_1 \cup \ldots \cup P_{k-1}) \cap L|$ $\leq r-1 < k-1$. Therefore, 
    \begin{align*}
        \mms^{k-r}_{v_i}\big((P_1 \cup \ldots \cup P_k) \setminus G_4\big) &\geq 
        \mms^{(k-1)-(r-1)}_{v_i}\big((P_1 \cup \ldots \cup P_{k-1}) \setminus G_4\big) \\
        &\geq 1-4\epsilon. \tag{by induction assumption}
    \end{align*}
    \paragraph{\boldmath Case 2: There exist $j, \ell \in [k]$, such that $|P_j \cap H| \geq 2$ and $|P_\ell \cap L| \geq 2$.}
    Similar to the former case, we have
    \begin{align*}
        \mms^{k-r}_{v_i}\big((P_1 \cup \ldots \cup P_k) \setminus G_4\big) &\geq 
        \mms^{(k-2)-(r-2)}_{v_i}\big((P_1 \cup \ldots \cup P_{k-2}) \setminus G_4\big) \\
        &\geq 1-4\epsilon. \tag{by induction assumption}
    \end{align*}
    
    \paragraph{\boldmath Case 3: Neither Case 1 nor Case 2 holds.}
    For all $j \in [k]$, we have $P_j \cap H = \emptyset$ or $P_j \cap L = \emptyset$; otherwise, we are in case 1.
    Let $S_1 := \{j \in [k] \mid P_j \cap L \neq \emptyset\}$ and $S_2 = [k] \setminus S_1 = \{j \in [k] \mid P_j \cap H \neq \emptyset\}$. If there exist bundles $P_j$ and $P_\ell$ such that $|P_j \cap H| \geq 2$ and $|P_\ell \cap L| \geq 2$, we are in case 2. Therefore, for all $j \in S_1$, $|P_j \cap L|=1$ or for all $j \in S_2$, $|P_j \cap H|=1$. Hence, there are $r$ bundles $P_1, \ldots, P_r$ such that either $|P_j \cap H| = 1$ (and $|P_j \cap L| = 0$) for all $j \in [r]$ or $|P_j \cap L| = 1$ (and $|P_j \cap H| = 0$) for all $j \in [r]$.
    
    \subparagraph{\boldmath Case 3.1: $k>r+1$.} 
    Assume $|P_j \cap H| = 1$ for all $j \in [r]$. (The case where $|P_j \cap L| = 1$ for all $j \in [r]$ is symmetric when $k>r+1$.)
    Let $|P_k \cap L| = a$. Then $|(P_1 \cup \ldots \cup P_a \cup P_k) \cap H| = a$ and $|(P_1 \cup \ldots \cup P_a \cup P_k) \cap L| = a$. Thus by the induction assumption, we have
    \begin{align*}
        \mms^{(a+1)-a}_{v_i}\big((P_1 \cup \ldots \cup P_a \cup P_k) \setminus G_4\big)
        &\geq 1-4\epsilon. 
    \end{align*}
    Moreover, $|(P_{a+1} \cup \ldots \cup P_{k-1}) \cap H| \leq r-a$ and $|(P_{a+1} \cup \ldots \cup P_{k-1}) \cap L| \leq r-a$. Thus by the induction assumption, we have
    \begin{align*}
        \mms^{(k-a-1)-(r-a)}_{v_i}\big((P_{a+1} \cup \ldots \cup P_{k-1}) \setminus G_4\big) 
        &\geq 1-4\epsilon. 
    \end{align*}
    So we can partition $(P_1 \cup \ldots \cup P_a \cup P_k) \setminus G_4$ into one bundle of value at least $1-4\epsilon$ for $i$ and also we can partition $(P_{a+1} \cup \ldots \cup P_{k-1}) \setminus G_4$ into $k-r-1$ bundles of value at least $1-4\epsilon$ for $i$. Thus, the lemma holds.
    \subparagraph{\boldmath Case 3.2: $k=r+1$.}
    Let $B = (P_1 \cup \ldots \cup P_k) \setminus G_4$. We want to show $\mms^1_{v_i}(B) \geq 1-4\epsilon$. Hence it suffices to show $v_i(B) \geq 1-4\epsilon$.
    \begin{align*}
        v_i(B) &\geq \sum_{j \in [k-1]} v_i\left(P_j \setminus (H \cup L)\right) \\
        &= \sum_{j \in [k-1]} \left(v_i(P_j) - v_i(P_j \cap (H \cup L))\right) \\
        &> (k-1)\left(1 - (\frac{3}{4}+\epsilon)\right) &\tag{since $|P_j \cap (H \cup L)|=1$, $v_i(P_j \cap (H \cup L)) \leq \frac{3}{4}+\epsilon$}\\
        &= (k-1)(\frac{1}{4}-\epsilon) \geq 1-4\epsilon. &\tag{for $k>4$}
    \end{align*}
    It remains to prove the claim when $k=3$ and $k=4$. If there are two bundles $P_1$ and $P_2$ such that $|P_1 \cap L| = |P_2 \cap L| = 1$, $v_i(B) \geq v_i(P_1 \setminus L) + v_i(P_2 \setminus L) > 2\left(1 - (\frac{1}{4}+\frac{\epsilon}{3})\right) > 1-4\epsilon$. Otherwise, for $k=3$, there are two bundles $P_1$ and $P_2$ such that $|P_1 \cap H|=|P_2 \cap H|=1$ and $|P_3 \cap L|=2$. Then, 
    \begin{align*}
        v_i(B) &= v_i(P_1 \setminus H) + v_i(P_2 \setminus H) + v_i(P_3 \setminus L) \\
        &> 2\left(1 - (\frac{3}{4}+\epsilon)\right) + \left(1 - 2(\frac{1}{4}+\frac{\epsilon}{3})\right) \\
        &= 1 - \frac{8\epsilon}{3} > 1-4\epsilon.
    \end{align*}
    For $k=4$, we have $|P_1 \cap H|=|P_2 \cap H|=|P_3 \cap H|=1$ and $|P_4 \cap L|=3$. Then, 
    \begin{align*}
        v_i(B) &= v_i(P_1 \setminus H) + v_i(P_2 \setminus H) + v_i(P_3 \setminus H) + v_i(P_4 \setminus L) \\
        &> 3\left(1 - (\frac{3}{4}+\epsilon)\right) + \left(1 - 3(\frac{1}{4}+\frac{\epsilon}{3})\right) = 1-4\epsilon.
\qedhere    \end{align*}
\end{proof}
We are ready to prove \cref{thm:1} and \cref{thm:2}.
\theoremOne*
\begin{proof}
    Fix an agent $i \in N'$. Let $\I^1$ be the instance after all applications of $R_1$ and before any further reduction. By \cref{valid-rules}, $\mms_i(\I^1) \geq 1$. So without loss of generality, let us assume $\I = \I^1$.
    Let $G_2$, $G_3$, and $G_4$ be the set of goods removed by applications of $R_2$, $R_3$, and $R_4$, respectively. Also, let $r_2 = |G_2|/3$, $r_3 = |G_3|/4$, and $r_4 = |G_4|/2$ be the number of times each rule is applied, respectively. By \cref{long-lemma}, $\mms^{n-r_4}_{v_i}(M \setminus G_4) \geq 1-4\epsilon$. For an application of $R_3$ (or $R_4$) at step $t$, let $\{a_1, a_2, a_3\}$ (or $\{b_1, b_2, b_3, b_4\}$) be the set of goods that are removed. By \cref{redund}, removing this set at a step $t' \geq t$ is still a valid reduction for $i$. Therefore, removing $G_2$ and $G_3$ and $r_2 + r_3$ agents does not decrease the MMS value of $i$. Thus, $\mms_i(\I') \geq 1-4\epsilon$.
\end{proof}

\theoremTwo*
\begin{proof}
    In $\mathtt{reduce}$, as long as $R^{(3/4+\epsilon)}_1$ is applicable, we apply it. Once it is not applicable anymore, for all remaining agents $i$, $v_i(1)<3/4+\epsilon$. In the rest of the procedure $\mathtt{reduce}$, we do not increase the value of any good for any agent. Therefore, $R^{(3/4+\epsilon)}_1$ remains inapplicable. As long as one of the rules $R^{(3/4+\epsilon)}_k$ is applicable for $k \in \{2,3,4\}$, we apply it. Therefore, $\mathtt{reduce}(\I, \epsilon)$ is $(3/4+\epsilon)$-irreducible. Let $\I' = (N',M',\mathcal{V}') = \mathtt{reduce}(\I, \epsilon)$. Since $\mms_i(\I') \geq 1-4\epsilon$ (by \cref{thm:1}), $\mathtt{normalize}$ can increase the value of each good by a multiplicative factor of at most $1/(1-4\epsilon)$. Therefore, after ordering the instance, none of the rules $R^\alpha_k$ for $k \in [4]$ would be applicable for $\alpha \geq \frac{3/4+\epsilon}{1-4\epsilon} = \frac{3}{4}+\frac{4\epsilon}{1-4\epsilon}$. Hence, $\hat{\I}=\mathtt{order}(\mathtt{normalize}(\mathtt{reduce}(\I, \epsilon)))$ is $\alpha$-irreducible for $\alpha \geq \frac{3}{4}+\frac{4\epsilon}{1-4\epsilon}$ and it is of course ordered. Since $\mathtt{order}$ does not change the multiset of the values of the goods for each agent, the instance remains normalized.

    Now let us assume $A$ is an $\alpha$-MMS allocation for $\hat{\I} = \mathtt{order}(\mathtt{normalize}(\mathtt{reduce}(\I, \epsilon)))$. By \cref{order-preserves}, we can obtain an allocation $B$ which is $\alpha$-MMS for $\mathtt{normalize}(\mathtt{reduce}(\I, \epsilon))$. \cref{thm:normalize} implies that $B$ is $\alpha$-MMS for $\I' = (N', M', \mathcal{V}') = \mathtt{reduce}(\I, \epsilon)$. For all agents $i \in N \setminus N'$, $v'_i(B_i) = v_i(B_i)/\mms_i(\I)$. Therefore, 
    \begin{align*}
        v_i(B_i) &= v'_i(B_i)\mms_i(\I) \\
        &\geq \alpha\mms_i(\I')\mms_i(\I) \tag{$B$ is $\alpha$-MMS for $\I'$}\\
        &\geq \alpha(1-4\epsilon)\mms^n_{v_i}(M). \tag{$\mms^n_{v'_i}(M) \geq 1-4\epsilon$ by \cref{thm:1}}
    \end{align*}
    Thus, $B$ gives all the agents in $N'$, $\alpha(1-4\epsilon)$ fraction of their MMS. All agents in $N \setminus N'$ receive $(3/4+\epsilon)$ fraction of their MMS value. Therefore, the final allocation is a $\min(3/4+\epsilon, (1-4\epsilon)\alpha)$-MMS allocation of $\I$.  
\end{proof}

\section{\boldmath $(3/4+\delta)$-MMS allocation for $\delta$-ONI instances}\label{easy-sec}
In this section, we prove that for $\delta \leq \lbdelta$ there exists a $(3/4+\delta)$-\mms~allocation if the input is a $\delta$-ONI instance. First we prove that in any $\delta$-ONI instance $\mathcal{I}=([n],[m],\mathcal{V})$, $m \geq 2n$.
\begin{observation}\label{2n-goods}
    For any $\delta \leq 1/4$, if $\I = ([n], [m], \mathcal{V})$ is $\delta$-ONI, then $m \geq 2n$.
\end{observation}
\begin{proof}
    Towards a contradiction, assume $m < 2n$. Now for an arbitrary agent $i$, let $(P_1, P_2, \ldots, P_n)$ be the MMS partition of $i$. Since $m < 2n$, there must be a bundle $P_j$ such that $|P_j| = 1$. Therefore, $v_i(1) \geq v_i(P_j)=1$ which means $R^{3/4+\delta}_1$ is applicable. This contradicts $\I$ being $(3/4+\delta)$-irreducible. Thus, $m \geq 2n$.
\end{proof}
We initialize $n$ bags $\{B_1, \ldots, B_n\}$ with the first $2n$ goods as follows:
\begin{equation}\label{eq:B_i}
    B_k := \{k , 2n-k+1\} \text{ for } k\in [n].
\end{equation}
See Figure \ref{B-bags} for a better intuition. Note that by \cref{2n-goods}, $m \geq 2n$ and such bag-initialization is possible.

Given an instance $\I = ([n],[m],\mathcal{V})$ (with $m \geq 2n$), let $N^1(\I) = \{i \in [n] \mid \forall k \in [n]: v_i(B_k) \leq 1\}$ and $N^2(\I) = \{i \in [n] \mid \exists k \in [n]: v_i(B_k) > 1\}$. 
\begin{observation}\label{obs:1-12}
For $\delta \leq 1/4$ and instance $\I$, if $\I$ is $\delta$-ONI, then for all agents $i \in N^2(\I)$, $v_i(2n+1) < 1/12+\delta$.
\end{observation}
\begin{proof}
    By the definition of $N^2$, there exist $k \in [n]$ such that $v_i(k) + v_i(2n-k+1) = v_i(B_k) > 1$. Therefore, by Lemma \ref{upper-13}, $v_i(k) > 2/3$. We have, 
    \begin{align*}
        v_i(2n+1) &< \frac{3}{4}+\delta - v_i(1) \tag{$R^{3/4+\delta}_4$ is not applicable}\\
        &\leq \frac{3}{4}+\delta - v_i(k) \tag{$v_i(1) \geq v_i(k)$}\\
        &< \frac{3}{4}+\delta - \frac{2}{3} = \frac{1}{12} + \delta, \tag{$v_i(k) > \frac{2}{3}$}
    \end{align*}
which completes the proof. 
\end{proof}
We refer to $N^1(\I)$ and $N^2(\I)$ by $N^1$ and $N^2$ respectively when $\I$ is the initial $\delta$-ONI instance. Recall that $N^1$ and $N^2$ do not change over the course of our algorithm.
Let $N^1_1 = \{i \in N^1 \mid v_i(2n+1) \geq 1/4 - 5\delta\}$ and $N^1_2 = N^1 \setminus N^1_1$.
Depending on the number of agents in $N^1_1$, we run one of the $\mathtt{approxMMS1}(\I, \delta)$ or $\mathtt{approxMMS2}(\I, \delta)$ shown in Algorithms \ref{algo:gt} or \ref{algo:new} respectively. Roughly speaking, if the size of $N^1_1$ is not too large, we run \cref{algo:gt} and prioritize agents in $N^1_1$. Otherwise, we run \cref{algo:new} giving priority to agents in $N^1_2 \cup N^2$. Giving priority to agents in a certain set $S$ means that when the algorithm is about to allocate a bag $B$ to an agent, if there is an agent in $S$ who gets satisfied upon receiving $B$ (i.e., $v_i(B) \geq 3/4+\delta$ for some $i \in S$), then the algorithms give $B$ to such an agent and not to someone outside $S$.

\subsection{\boldmath Case $1$: $|N^1_1| \leq \lbnone$}\label{small-N11}
In this case we run \cref{algo:gt}.
\begin{algorithm}[!tb]
    \caption{$\mathtt{approxMMS1}(\mathcal{I}, \delta)$}
    \label{algo:gt}
    \textbf{Input:} $\delta$-ONI $\mathcal{I}=(N,M,\mathcal{V})$ and factor $\delta$

    \textbf{Output:} Allocation $A = \langle A_1, \ldots, A_n \rangle$
    \begin{algorithmic}
        \State $B_i \leftarrow \{i, 2n-i+1\}_{i \in [n]}$
        \State $\mathcal{B} = \cup_{i \in [n]} \{B_i\}$
        \State $\alpha = 3/4 + \delta$
        \While{$\exists i \in N, B \in \mathcal{B}$ s.t. $v_i(B) \geq \alpha$}
            \State $i \leftarrow$ an arbitrary agent s.t. $v_i(B) \geq \alpha$, priority with agents in $N^1_1$
            \State $A_i \leftarrow B$
            \State $\mathcal{B} \leftarrow \mathcal{B} \setminus \{B\}$
            \State $N \leftarrow N \setminus \{i\}$
            \State $M \leftarrow M \setminus B$
        \EndWhile
        \State $J \leftarrow \cup_{B \in \mathcal{B}} B$
        \For {$B \in \mathcal{B}$}
            \While {$\nexists i \in N$ s.t. $v_i(B) \geq \alpha$}
                \State $g \leftarrow$ an arbitrary good in $M \setminus J$
                \State $B \leftarrow B \cup \{g\}$
                \State $M \leftarrow M \setminus \{g\}$
            \EndWhile
            \State $i \leftarrow$ an arbitrary agent s.t. $v_i(B) \geq \alpha$, priority with agents in $N^1_1$
            \State $A_i \leftarrow B$
            \State $N \leftarrow N \setminus \{i\}$
            \State $M \leftarrow M \setminus B$            
        \EndFor
        \State \Return $\langle A_1, \ldots, A_n \rangle$
    \end{algorithmic}
\end{algorithm}
For $k\in [n]$, let $B_k$ and $\hat{B}_k\supseteq B_k$ be the $k^{th}$ bag at the beginning and end of \cref{algo:gt}, respectively.
\begin{lemma}\label{upper-bound-1}
    Let $i$ be any agent who did not receive any bag by the end of \cref{algo:gt}.  For all $k \in [n]$ such that $v_i(B_k) \leq 1$, we have $v_i(\hat{B}_k) < 1 + 4\delta/3$.
\end{lemma}
\begin{proof}
    The claim trivially holds if $\hat{B}_k = B_k$. Now assume $B_k \subsetneq \hat{B}_k$.
    Let $g$ be the last good added to $\hat{B}_k$.
    We have $v_i(\hat{B}_k \setminus g) < 3/4 +\delta$, otherwise $g$ would not be added to $\hat{B}_k$.
    Also note that $g > 2n$ and hence $v_i(g) < 1/4 + \delta/3$
    by \cref{thm:vr-upper-bounds}. Thus, we have
    \begin{align*}
        v_i(\hat{B}_k) &= v_i(\hat{B}_k \setminus g) + v_i(g) \\
        &< \left(\frac{3}{4}+\delta\right) + \left(\frac{1}{4}+\frac{\delta}{3} \right) = 1+\frac{4\delta}{3}.
\qedhere    \end{align*}
\end{proof}

\begin{lemma}\label{lemma-1}
    For $\delta \leq \frac{1}{4}$, given a $\delta$-ONI instance with $|N^1_1| \leq \lbnone$, all agents $i \in N^1_1$ receive a bag of value at least $(3/4 +\delta) \cdot \mms_i$ at the end of \cref{algo:gt}.
\end{lemma}
\begin{proof}
    It suffices to prove that all agents $i \in N^1_1$ receive a bag at the end of \cref{algo:gt}. Towards a contradiction, assume that $i \in N^1_1$ does not receive any bag.
    \begin{claim}\label{claim-1-1}
        For all bags $B$ not allocated to an agent in $N^1_1$, $v_i(B) < 3/4 +\delta$. 
    \end{claim}
    \cref{claim-1-1} holds since the priority is with agents in $N^1_1$. Let $S$ be the set of bags allocated to agents in $N^1_1$ and $\bar{S}$ be the set of the remaining bags.
    We have
    \begin{align*}
        v_i(M) &= \sum_{k \in [n]} v_i(\hat{B}_k) = \sum_{B \in S} v_i(B) + \sum_{B \in \bar{S}} v_i(B) \\
        &< |N^1_1|\left(1 + \frac{4\delta}{3}\right) + \left(n - |N^1_1|\right)\left(\frac{3}{4}+\delta\right) &\tag{\cref{upper-bound-1} and \cref{claim-1-1}} \\
        &\leq n, &\tag{$|N^1_1| \leq \lbnone$}
    \end{align*}
    which is a contradiction since $v_i(M)=n$. Thus, all agents $i \in N^1_1$ receive a bag at the end of \cref{algo:gt}.
\end{proof}
\begin{remark}
The last inequality in the proof of \cref{lemma-1} is tight for $|N^1_1| = \lbnone$. 
\end{remark}

\begin{lemma}\label{lemma-2}
    For $\delta \leq \frac{1}{4}$, given a $\delta$-ONI instance with $|N^1_1| \leq \lbnone$, all agents $i \in N^1_2$ receive a bag of value at least $(3/4+\delta) \cdot \mms_i$ at the end of \cref{algo:gt}.
\end{lemma}
\begin{proof}
    It suffices to prove that all agents $i \in N^1_2$ receive a bag at the end of \cref{algo:gt}. Towards a contradiction, assume that $i \in N^1_2$ does not receive any bag.
    \begin{claim}\label{claim-2-1}
        For all $k \in [n]$, $v_i(\hat{B}_k) \leq 1$. 
    \end{claim}
    \begin{claimproof}
    The claim trivially holds if $\hat{B}_k = B_k$. Now assume $B_k \subsetneq \hat{B}_k$.
    Let $g$ be the last good added to $\hat{B}_k$.
    We have $v_i(\hat{B}_k \setminus g) < 3/4+\delta$, otherwise $g$ would not be added to $\hat{B}_k$.
    Also note that $g \geq 2n+1$ and hence $v_i(g) \leq v_i(2n+1) < 1/4 - 5 \delta$
    by the definition of $N^1_2$. Therefore, we have
    \begin{align*}
        v_i(\hat{B}_k) &= v_i(\hat{B}_k \setminus g) + v_i(g) \\
        &< (\frac{3}{4}+\delta) + (\frac{1}{4} - 5 \delta) < 1.
    \end{align*}
    Thus, the claim holds. 
    \end{claimproof}
    Since agent $i$ did not receive a bag, there exists an unallocated bag with value less than $1$ for agent $i$. Therefore, $v_i(M) = \sum_{k \in [n]} v_i(\hat{B}_k) <n$ which is a contradiction. Thus, all agents $i \in N^1_2$ receive a bag at the end of \cref{algo:gt}.
\end{proof}

\subsubsection{\boldmath Agents in $N^2$}
In this section, we prove that all agents in $N^2$ also receive a bag at the end of \cref{algo:gt}. For the sake of contradiction, assume that agent $i \in N^2$ does not receive a bag at the end of \cref{algo:gt}. Let $A^+ := \{k \in [n] \mid v_i(B_k) > 1\}$ and $A^- := \{k \in [n] \mid v_i(B_k) < 3/4 +\delta\}$ be the indices of the bags satisfying the respective constraint. Also, let $\ell$ be the smallest such that for all $k \in [\ell+1, n]$, $v_i(k) + v_i(2n-k+1 + \ell) < 1$. See Figure \ref{fig:ell} taken from \cite{simple}. 
\begin{figure*}[tb]
\centering
\newlength{\cellW}\newlength{\cellH}
\setlength{\cellW}{3.45em}
\setlength{\cellH}{1.8em}
\begin{tikzpicture}[
outerBox/.style = {semithick},
innerBord/.style = {},
highlight/.style = {very thick},
myArrow/.style={->,>={Stealth},thick},
]
\draw[outerBox] (0, 0) rectangle +(12\cellW, 2\cellH);
\draw[innerBord] (0, 1\cellH) -- +(12\cellW, 0);
\foreach \x in {1,2,3.5,5.5,7,8.5,10,11}
    \draw[innerBord] (\x\cellW, 0) -- +(0, 2\cellH);
\draw[highlight] (3.5\cellW, 1\cellH) rectangle +(2\cellW, 1\cellH);
\draw[highlight] (7\cellW, 0) rectangle +(1.5\cellW, 1\cellH);
\foreach \x/\w/\downText/\upText in {
        0/1/1/2n,
        1/1/2/2n-1,
        2/1.5/\cdots/\cdots,
        3.5/2/k-\ell/2n+1-k+\ell,
        5.5/1.5/\cdots/\cdots,
        7/1.5/k/2n+1-k,
        8.5/1.5/\cdots/\cdots,
        10/1/n-1/n+2,
        11/1/n/n+1
        } {
    \path (\x\cellW, 0) rectangle +(\w\cellW, 1\cellH) node[pos=0.5] {$\downText$};
    \path (\x\cellW, 1\cellH) rectangle +(\w\cellW, 1\cellH) node[pos=0.5] {$\upText$};
}
\node[circle,draw,inner sep=0] (plus) at (6.5\cellW, -0.6\cellH) {$+$};
\draw[thick] (5\cellW, 1\cellH) -- (plus);
\draw[thick] (7.75\cellW, 0) -- (plus);
\draw[->, thick] (plus) -- +(0, -0.8\cellH);
\node (le1) at (6.5\cellW, -1.7\cellH) {$\le 1$};
\end{tikzpicture}
\caption{The items $[2n]$ are arranged in a table, where the $k^{\text{th}}$ column is $B_k = \{k, 2n+1-k\}$.
$\ell$ is the smallest \emph{shift} such that $v_i(k) + v_i(2n+1-k+\ell) \leq 1$ for all $k$.}
\label{fig:ell}
\end{figure*}
\cite{simple} proved $\sum_{k \in A^+} v_i(\hat{B}_k) < |A^+| + \ell(\frac{1}{12}+\delta)$. For completeness, we repeat its proof in Appendix~\ref{app:mp}.

\begin{restatable}{lemma}{trickylemma}\cite{simple}\label{tricky-bound}
    $\sum_{k \in A^+} v_i(\hat{B}_k) < |A^+| + \ell(\frac{1}{12}+\delta)$.
\end{restatable}

\begin{observation}\label{upper-56}
    For all $k \in A^-$, $v_i(\hat{B}_k) < \frac{5}{6} + 2\delta$.
\end{observation}
\begin{proof}
    If $\hat{B}_k = B_k$, then $v_i(\hat{B}_k) < 3/4 + \delta < 5/6 + 2\delta$. Otherwise, let $g$ be the last good added to $\hat{B}_k$. Note that $v_i(\hat{B}_k \setminus g) < 3/4 + \delta$, otherwise the algorithm would assign $\hat{B}_k \setminus g$ to agent $i$ instead of adding $g$ to it. We have
    \begin{align*}
        v_i(\hat{B}_k) &= v_i (\hat{B}_k \setminus g) + v_i(g) \\
        &< (\frac{3}{4}+\delta) + v_i(2n+1) &\tag{$v_i (\hat{B}_k \setminus g) < \frac{3}{4} +\delta$ and $v_i(g) \leq v_i(2n+1)$} \\
        &< (\frac{3}{4}+\delta) + (\frac{1}{12}+\delta)= \frac{5}{6} + 2\delta. &\tag{$v_i(2n+1) < \frac{1}{12} +\delta$ by \cref{obs:1-12}}
    \end{align*}
\end{proof}
\begin{observation} \label{lower-12}
    For all $k \in [n]$, $v_i(B_k) > \frac{1}{2} - 2\delta$.
\end{observation}
\begin{proof}
    Let $t$ be smallest such that $v_i(B_t) > 1$. By \cref{upper-13}, $v_i(t) > \frac{2}{3}$. Therefore, for all $k \leq t$,
    \begin{align*}
        v_i(B_k) \geq v_i(k) \geq v_i(t) > \frac{2}{3} > \frac{1}{2}.
    \end{align*}

    Note that $v_i(t) + v_i(2n-t+1) > 1$ and by \cref{thm:vr-upper-bounds}, $v_i(t) < 3/4+\delta$. Thus, $v_i(2n-t+1) > 1/4-\delta$. For all $k > t$, we have
    \begin{align*}
        v_i(B_k) &= v_i(k) + v_i(2n-k+1) \\
        &\geq 2 \cdot v_i(2n-t+1) &\tag{$k < 2n-k+1 \leq 2n-t+1$}\\
        &> \frac{1}{2} - 2\delta. &\tag{$v_i(2n-t+1) > \frac{1}{4}-\delta$}
    \end{align*}  
\end{proof}
\begin{observation}\label{small-goods-bound}
    $v_i(M \setminus [2n]) > \ell(\frac{1}{4}-\delta)$.
\end{observation}
\begin{proof}
    By the definition of $\ell$, there exists a $k \in \{\ell, \ldots, n\}$ such that $v_i(k) + v_i(2n-k+\ell) > 1$. Therefore, for all $j \leq k$ and $t \leq 2n-k+\ell$, $v_i(j) + v_i(t) > 1$. Let $P = (P_1, \ldots, P_n)$ be an \mms~partition of agent $i$. For $j \in [k]$, let $j \in P_j$. Note that for different $j,j' \in [k]$, $P_j$ and $P_{j'}$ are different since $v_i(j) + v_i(j') > 1 =v_i(P_j)$. Also note that for every good $g \in [2n-k+\ell]$ and $j \in [k]$, $g \notin P_j$, otherwise $v_i(P_j)>1$. Therefore, there are at least $\ell$ bundles like $P_j$ among $P_1, \ldots, P_k$ such that $P_j \cap [2n] = \{j\}$. We have
    \begin{align*}
        v_i(M \setminus [2n]) &\geq \sum_{j \in [k]} v_i(P_j \setminus \{j\}) \geq \sum_{j \in [\ell]} \left( v_i(P_j) - v_i(j) \right) \\
        &> \sum_{j \in [\ell]} \left(1 - (\frac{3}{4}+\delta)\right) = \ell(\frac{1}{4}-\delta).
        &\tag{$v_i(j) < \frac{3}{4}+\delta$ by \cref{thm:vr-upper-bounds}}
    \end{align*}
\end{proof}
We are now ready to prove \cref{lemma-3}.
\begin{lemma}\label{lemma-3}
    For $\delta \leq 0.011$, given a $\delta$-ONI instance with $|N^1_1| \leq \lbnone$, all agents $i \in N^2$ receive a bag of value at least $(\frac{3}{4}+\delta)$ at the end of \cref{algo:gt}.
\end{lemma}
\begin{proof}
    It suffices to prove that all agents $i \in N^2$ receive a bag at the end of \cref{algo:gt}. Towards a contradiction, assume that $i \in N^2$ does not receive any bag. For all $k \in N \setminus (A^- \cup A^+)$, since $v_i(B_k) \ge 3/4 + \delta$ and $i$ has not received a bag, $\hat{B}_k = B_k$. Thus, for all $k \in N \setminus (A^- \cup A^+)$
    \begin{align}
        v_i(\hat{B}_k) = v_i(B_k) \leq 1. \label{ineq-easy}
    \end{align}
    We have
    \begin{align*}
        n &= v_i(M) = \sum_{k \in A^-} v_i(\hat{B}_k) + \sum_{k \in A^+} v_i(\hat{B}_k) + \sum_{k \in N \setminus (A^- \cup A^+)} v_i(\hat{B}_k)\\
        &< \left(|A^-| (\frac{5}{6} + 2\delta)\right) + \left(|A^+| + \ell(\frac{1}{12}+\delta)\right) + \left(n - |A^-| - |A^+|\right) &\tag{\cref{upper-56}, \cref{tricky-bound} and Inequality \eqref{ineq-easy}} \\
        &= n - |A^-|(\frac{1}{6} - 2\delta) + \ell(\frac{1}{12}+\delta). 
    \end{align*}
    Therefore, we have
    \begin{align}
        \frac{|A^-|}{\ell} &< \frac{1/12+\delta}{1/6 - 2\delta} \label{ineq1}
    \end{align}
    Next, we bound the value of the goods in $M \setminus [2n]$ and contradict Inequality \eqref{ineq1}.
    We have,
    \begin{align*}
        \ell(\frac{1}{4} - \delta) &\leq v_i(M \setminus [2n]) &\tag{\cref{small-goods-bound}}\\
        &= \sum_{k \in A^-} \left(v_i(\hat{B}_k) - v_i(B_k)\right) &\tag{$M \setminus [2n] = \bigcup_{k \in A^-} (\hat{B}_k \setminus B_k)$}\\
        &< |A^-|\left((\frac{5}{6}+\delta) - (\frac{1}{2}-2\delta)\right) &\tag{\cref{upper-56} and \cref{lower-12}}\\
        &= |A^-| \cdot (\frac{1}{3} + 3\delta).
    \end{align*}
    Thus, 
    \begin{align}
        \frac{|A^-|}{\ell} &> \frac{1/4 - \delta}{1/3 + 3\delta} \label{ineq2}
    \end{align}
    Inequalities \eqref{ineq1} and \eqref{ineq2} imply that $\frac{1/12+\delta}{1/6 - 2\delta}>\frac{1/4 - \delta}{1/3 + 3\delta}$, which is a contradiction with $\delta \leq 0.011$. Thus, all agents $i \in N^2$ receive a bag at the end of \cref{algo:gt}.
\end{proof}

\begin{theorem}\label{thm:3}
    Given any $\delta \leq 0.011$, for all $\delta$-ONI instances where $|N^1_1| \leq \lbnone$, \cref{algo:gt} returns a $(\frac{3}{4}+\delta)$-\mms~allocation.
\end{theorem}
\begin{proof}
    Since $N = N^1_1 \cup N^1_2 \cup N^2$, by Lemmas \ref{lemma-1}, \ref{lemma-2} and \ref{lemma-3} all agents receive a bag of value at least $(\frac{3}{4}+\delta) \cdot \mms_i$ in \cref{algo:gt}.
\end{proof}    

\subsection{\boldmath Case $2$: $|N^1_1| > \lbnone$}\label{large-N11}
In this case, we run Algorithm \ref{algo:new}. Starting from an ordered normalized $(3/4 + \delta)$-irreducible instance, as long as there is a bag $B_k$ with value at least $3/4+\delta$ for some agent, we give $B_k$ to such an agent. The priority is with agents who initially belonged to $N^1_2 \cup N^2$. Therefore, in the remaining instance, all bags are of value less than $3/4+\delta$ for all the remaining agents.
We introduce one more reduction rule in this section. 
\begin{itemize}
    \item \textbf{$R^\alpha_5:$} If $v_i(1) + v_i(2) \geq \alpha$ for some $i \in N$, allocate $\{1,2\}$ to agent $i$ and remove $i$ from $N$. The priority is with agents in $N^1_2 \cup N^2$.
\end{itemize}
Starting from an ordered normalized $(3/4 + \delta)$-irreducible instance, after allocating bags of value at least $3/4 +\delta$ to some agents, we run $R^{3/4+\delta}_5$ as long as it is applicable. For ease of notation, we write $R_j$ instead of $R^{3/4+\delta}_j$ for $j \in [5]$. Then, we run $R_2$ and $R_3$ as long as they are applicable. 
Afterwards, for all $k \in [n]$, we initialize $C_k = \{k, 2n-k+1, 2n+k\}$.\footnote{Note that it is without loss of generality to assume $m \geq 3n$. If $m < 3n$, add dummy goods of value $0$ to everyone. The MMS value of the agents remains the same, and any $\alpha$-MMS allocation in the final instance is an $\alpha$-MMS allocation in the original instance after removing the dummy goods.} See Figure \ref{C-bags} for better intuition. Then, we do bag-filling. Let $\hat{C}_k$ be the result of bag-filling on bag $C_k$. The pseudocode of this algorithm is shown in \cref{algo:new}. 
\begin{figure}[t]\centering
\begin{tikzpicture}
[scale=1,
 good/.style={circle, draw=black, thick, minimum size=30pt},
]

\draw[black, very thick] (-0.4-0.25,0.8) rectangle (0.4+0.25,4.7);

\node[good]      at (0,4)      {$\scriptstyle{2n+1}$};
\node[good]      at (0,2.75)      {$\scriptstyle{2n}$};
\node[good]      at (0,1.5)      {$\scriptstyle{1}$};

\node at (0, 0.5) {$\scriptstyle{C_1}$};

\filldraw[color=black!60, fill=black!5, thick](1.6, 2.7) circle (0.02);
\filldraw[color=black!60, fill=black!5, thick](1.7, 2.7) circle (0.02);
\filldraw[color=black!60, fill=black!5, thick](1.8, 2.7) circle (0.02);

\draw[black, very thick] (-1.4-0.25 +4.5,0.8) rectangle (-0.6+0.25 +4.5,4.7);

\node[good]      at (0 +3.5,4)      {$\scriptstyle{2n+k}$};
\node[good, scale=0.74]      at (0 +3.5,2.75)      {$\scriptstyle{2n-k+1}$};
\node[good]      at (0 +3.5,1.5)      {$\scriptstyle{k}$};

\node at (0 +3.5, 0.5) {$\scriptstyle{C_k}$};

\filldraw[color=black!60, fill=black!5, thick](0.1+5, 2.7) circle (0.02);
\filldraw[color=black!60, fill=black!5, thick](0.2+5, 2.7) circle (0.02);
\filldraw[color=black!60, fill=black!5, thick](0.3+5, 2.7) circle (0.02);

\draw[black, very thick] (-0.4-5.25 +12,0.8) rectangle (0.4+0.25-5 +12,4.7);

\node[good]      at (0 +7,4)      {$\scriptstyle{3n}$};
\node[good]      at (0 +7,2.75)      {$\scriptstyle{n+1}$};
\node[good]      at (0 +7,1.5)      {$\scriptstyle{n}$};

\node at (0 +7, 0.5) {$\scriptstyle{C_n}$};

\end{tikzpicture}
\caption{Configuration of Bags $C_1, C_2, \ldots, C_n$}
\label{C-bags}
\end{figure}
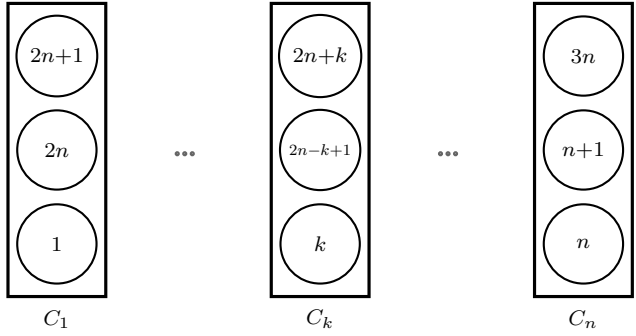

\begin{algorithm}[!htb]
    \caption{$\mathtt{approxMMS2}(\mathcal{I}, \delta)$}
    \label{algo:new}
    \textbf{Input:} $\delta$-ONI instance $\mathcal{I}=(N,M,\mathcal{V})$ and factor $\delta$

    \textbf{Output:} Allocation $A = \langle A_1, \ldots, A_n \rangle$
    \begin{algorithmic}
        \State $B_i \leftarrow \{i, 2n-i+1\}_{i \in [n]}$
        \State $\mathcal{B} = \cup_{i \in [n]} \{B_i\}$
        \State $\alpha = 3/4 +\delta$
        \While{$\exists i \in N, B \in \mathcal{B}$ s.t. $v_i(B) \geq \alpha$}
            \State $i \leftarrow$ an arbitrary agent s.t. $v_i(B) \geq \alpha$, priority with agents in $N^1_2 \cup N^2$
            \State $A_i \leftarrow B$
            \State $\mathcal{B} \leftarrow \mathcal{B} \setminus \{B\}$
            \State $N \leftarrow N \setminus \{i\}$
            \State $M \leftarrow M \setminus B$
        \EndWhile
        \While{$R^\alpha_5(\alpha)$ is applicable}
            \State apply $R^\alpha_5(\alpha)$
        \EndWhile
        \While{$R^\alpha_2$ or $R^\alpha_3$ is applicable}
            \State apply $R^\alpha_k$ for smallest $k \in \{2,3\}$ s.t. $R^\alpha_k$ is applicable
        \EndWhile
        \State $n \leftarrow |N|$
        \State $C_i \leftarrow \{i, 2n-i+1, 2n+i\}_{i \in [n]}$
        \For {$k \leftarrow 1$ to $n$}
            \While {$\nexists i \in N$ s.t. $v_i(C_k) \geq \alpha$}
                \State $g \leftarrow$ an arbitrary good in $M \setminus [3n]$
                \State $C_k \leftarrow C_k \cup \{g\}$
                \State $M \leftarrow M \setminus \{g\}$
            \EndWhile
            \State $i \leftarrow$ an arbitrary agent s.t $v_i(C_k) \geq \alpha$, priority with agents in $N^1_2 \cup N^2$
            \State $A_i \leftarrow C_k$
            \State $N \leftarrow N \setminus \{i\}$
            \State $M \leftarrow M \setminus C_k$            
        \EndFor
        \State \Return $\langle A_1, \ldots, A_n \rangle$
    \end{algorithmic}
\end{algorithm}
\begin{lemma}\label{bad-bags}
    For all agents $i \in N^1_2 \cup N^2$ and bags $B$ which is allocated to an agent in $N^1_2 \cup N^2$ during \cref{algo:new}, $v_i(B) < 3/2+2\delta$. 
\end{lemma}
\begin{proof}
    We prove the lemma by upper bounding the value of the bags allocated at each step.
    \begin{claim}\label{bad-bags-1}
        For all bags $B$ allocated to an agent before or during $R_5$, $v_i(B) < 3/2 +2\delta$. 
    \end{claim}
    \begin{claimproof}
    Since we start with a $(3/4 + \delta)$-irreducible instance, by \cref{thm:vr-upper-bounds}, for all goods $g$, $v_i(g) < 3/4 +\delta$. Therefore, for all the bags $B$ of size two, we have $v_i(B) < 3/2 +2\delta$.
    \end{claimproof}
    \begin{claim}\label{bad-bags-2}
        For all bags $B$ which is allocated to an agent during $R_2$, $v_i(B) < 3/2 +2\delta$. 
    \end{claim}
    \begin{claimproof}
    Note that when we run $R_2$, $R_5$ is not applicable. Therefore, $v_i(1) + v_i(2) < 3/4 +\delta$. Hence,
    $v_i(\{2n-1, 2n, 2n+1\}) \leq v_i(\{1,2\}) + v_i(2n+1) < 2(3/4+\delta) = 3/2 +2\delta$.
    \end{claimproof}
    \begin{claim}\label{bad-bags-3}
        For all bags $B$ which is allocated to an agent during $R_3$, $v_i(B) < 3/2 +2\delta$. 
    \end{claim}
    \begin{claimproof}
    Note that when we run $R_3$, $R_5$ is not applicable. Therefore, $v_i(1) + v_i(2) < 3/4 +\delta$. Hence,
    $v_i(\{3n-2, 3n-1, 3n, 3n+1\}) \leq 2v_i(\{1,2\}) < 3/2 +2\delta$.
    \end{claimproof}
    \begin{claim}\label{bad-bags-4}
        For all bags $B$ allocated to an agent during the bag-filling phase, $v_i(B) < 3/2 +2\delta$. 
    \end{claim}
    \begin{claimproof}
    If $B = \{k, 2n-k+1, 2n+k\}$, similar to the claims above, $v_i(B) \leq v_i(\{1,2\}) + v_i(2n+k) \leq 2(3/4 +\delta) = 3/2 +2\delta$. Otherwise, let $g$ be the last good added to $B$.
    We have $v_i(B \setminus g) < 3/4 +\delta$, otherwise $g$ would not be added to $B$.
    Therefore, we have $v_i(B) = v_i(B \setminus g) + v_i(g) < 2(3/4+\delta) = 3/2 +2\delta$.
    \end{claimproof}
    
    By Claims \ref{bad-bags-1}, \ref{bad-bags-2}, \ref{bad-bags-3} and \ref{bad-bags-4}, all bags that are allocated during \cref{algo:new} are of value less than $3/2 +2\delta$. Therefore, the lemma holds.
\end{proof}

\begin{lemma}\label{n12andn2}
    For $\delta \leq 1/20$, given a $\delta$-ONI instance with $|N^1_1| > \lbnone$, all agents in $N^1_2 \cup N^2$ receive a bag of value at least $3/4 +\delta$ at the end of \cref{algo:new}.
\end{lemma}
\begin{proof}
    It suffices to prove that all agents $i \in N^1_2 \cup N^2$ receive a bag at the end of \cref{algo:new}. Towards a contradiction, assume that $i \in N^1_2 \cup N^2$ does not receive any bag.
    \begin{claim}\label{cl1}
        For all bags $B$ which is either unallocated or is allocated to an agent in $N^1_1$, $v_i(B) < 3/4 +\delta$. 
    \end{claim}
    The claim holds since the priority is with agents in $N^1_2 \cup N^2$ and also that we allocate all the bags of value at least $3/4 + \delta$ for some remaining agent.
   
    Let $S$ be the set of bags allocated to agents in $N^1_2 \cup N^2$ and $\Bar{S}$ be the set of the remaining bags.
    We have
    \begin{align*}
        n &= v_i(M) = \sum_{B \in S} v_i(B) + \sum_{B \in \Bar{S}} v_i(B) \\
        &\leq (n - |N^1_1|)\left(\frac{3}{2}+2\delta\right) + |N^1_1|\left(\frac{3}{4}+\delta\right) &\tag{\cref{bad-bags} and \cref{cl1}}\\
        &= \left(\frac{3}{4}+\delta \right)(2n- |N^1_1|) \\
        &< n(\frac{3}{4}+\delta)(2- \frac{\frac{1}{4}-\delta}{\frac{1}{4}+\frac{\delta}{3}}). \tag{$|N^1_1|> \lbnone$ } \\
        &= 3n( \frac{5\delta}{3} + \frac{1}{4})
    \end{align*}
    This implies that $\frac{5\delta}{3} + \frac{1}{4}>\frac{1}{3}$. which is a contradiction with $\delta \leq  1/20$. Therefore, all agents $i \in N^1_2 \cup N^2$ receive a bag at the end of \cref{algo:new}.
\end{proof}
\subsubsection{\boldmath Agents in $N^1_1$}
In this section, we prove that all agents in $N^1_1$ also receive a bag at the end of \cref{algo:new}. First, we prove a general lemma that lower bounds the MMS value of an agent after allocating $2k$ goods to $k$ other agents. This way, we can lower bound the MMS value of agents in $N^1_1$ after the sequence of $R_5$ rules is applied.
\begin{lemma}\label{sequence-reduction}
    Given a set of goods $M$ and a valuation function $v$, let $S \subseteq M$ be such that $|S| = 2k$ for $k < n$ and $x \geq 0$ be such that $v(g) \leq \mms^n_v(M)/2 + x$ for all $g \in S$. Then, $\mms^{n-k}_v(M \setminus S) \geq \mms^n_v(M)-2x$.
\end{lemma}
\begin{proof}
    We construct a partition of a subset of $M \setminus S$ into $n-k$ bundles such that the minimum value of these bundles is at least $\mms^n_v(M)-2x$. Let $(P_1, \ldots, P_n)$ be an MMS partition of $M$ according to valuation function $v$. For all $j \in [n]$, let $Q_j = P_j \cap S$. Without loss of generality, assume $|Q_1| \geq \ldots \geq |Q_n|$. Let $t$ be largest such that for all $\ell \leq t$, $\sum_{j \in [\ell]} |Q_j| \geq 2\ell$. 
    This implies that $|Q_{t+1}| \leq 1$. 
    \begin{claim}\label{claim10}
        $\sum_{j \in [t]} |Q_j| = 2t$.
    \end{claim}
    \begin{claimproof}
    If $|Q_{t+1}|=1$, and $\sum_{j \in [t]} |Q_j| > 2t$, then $\sum_{j \in [t+1]} |Q_j| \geq 2(t+1)$ which is a contradiction with the definition of $t$. If $|Q_{t+1}|=0$, then $\sum_{j \in [t]} |Q_j| = 2k$. If $t<k$, then $\sum_{j \in [t+1]} |Q_j| =2k \geq 2(t+1)$ which is again a contradiction with the definition of $t$. So in this case, $t=k$ and therefore, $\sum_{j \in [t]} |Q_j| = 2t$. Hence, \cref{claim10} holds.
    \end{claimproof}
    \begin{claim}\label{claim11}
        $Q_{2k-t+1} = \emptyset$.
    \end{claim}    
    \begin{claimproof}
    If $Q_{2k-t+1} \neq \emptyset$ then $|Q_{2k-t+1}| \geq 1$. Therefore, 
    \begin{align*}
        \sum_{j \in [2k-t+1]} |Q_j| &= \sum_{j \leq t} |Q_j| + \sum_{t < j \leq 2k-t+1} |Q_j| \\
        &\geq 2t + (2k-2t+1) \tag{\cref{claim10}, and $|Q_j| \geq 1$ for $j \leq 2k-t+1$}\\
        &> k,
    \end{align*}
    which is a contradiction. Therefore, \cref{claim11} holds.    
    \end{claimproof}

    Now we remove the first $t$ bundles (i.e., $P_1, \ldots, P_t$) and merge the next $k-t$ pairs of bundles after removing $S$ (i.e., $(P_{t+1} \setminus S)$ with $(P_{t+2} \setminus S)$ and so on) as follows: 
    \begin{align*}
        \hat{P} = \big( (P_{t+1} \cup P_{t+2}) \setminus S, (P_{t+3} \cup P_{t+4}) \setminus S, \ldots, (P_{2k-t-1} \cup P_{2k-t}) \setminus S, P_{2k-t+1}, \ldots, P_n \big).
    \end{align*}
    \cref{claim11} implies that for all $j>2k-t$, $P_j = P_j \setminus S$. Therefore, $\hat{P}$ is a partition of the goods in $(M \setminus (P_1 \cup \ldots \cup P_t)) \setminus S \subseteq M \setminus S$. For all $j > 2k-t$, we have $v_i(P_j) \geq \mms^{|N|}_v(M) \geq \mms^{|N|}_v(M) - 2x$. Also, for all $t < j \leq 2k-t$, we have $|P_j \cap S| \leq 1$ and $v_i(g) \leq \mms^{|N|}_v(M)/2 +x$ for all $g \in P_j$. Therefore,
    \begin{align*}
        v_i(P_j \setminus S) &\geq v_i(P_j) - (\mms^{|N|}_v(M)/2 + x) \\
        &\geq \mms^{|N|}_v(M)/2 - x.
    \end{align*}
    Thus, for all $t < j < 2k-t$, $v_i((P_j \cup P_{j+1}) \setminus S) \geq \mms^{|N|}_v(M) - 2x$. Hence, $\hat{P}$ is a partition of a subset of $M \setminus S$ into $n-k$ bundles with minimum value at least $\mms^{|N|}_v(M) - 2x$. Therefore, \cref{sequence-reduction} holds.
\end{proof}
\begin{lemma}\label{mms-bound}
    Let $i \in N^1_1$ be a remaining agent after no more $R_5$ is applicable. Then, before applying more reduction rules, $\mms_i \geq 1 - 12\delta$. 
\end{lemma}
\begin{proof} 
    We start by proving the following claim.
    \begin{claim}\label{claim12}
        Right before applying any $R_5$, $v_i(1) \leq 1/2 +6\delta$.
    \end{claim}
    \begin{claimproof}
    Right before applying any $R_5$, no bag is of value at least $\frac{3}{4}+\delta$ to any agent and in particular agent $i$. Therefore, $v_i(1) + v_i(2n+1) \leq v_i(1) + v_i(2n) < 3/4 +\delta$. Since $v_i(2n+1) > 1/4 -5\delta$, by the definition of $R_5$, we have $v_i(1) < 1/2 +6\delta$. Therefore the claim holds.
    \end{claimproof}
    
    Consider the step right before applying any $R_5$. Note that until this step, only some $B_j$'s are allocated. Since $i \in N^1$, $v_i(B_j) \leq 1$ for all $j \in [n]$ and since $|B_j| = 2$, allocating $B_j$'s are valid reductions for agent $i$ by Lemma \ref{simple-valid}. Thus, before applying any $R_5$, $\mms_i \geq 1$. Now let $\I' = ([n'], M', \mathcal{V})$ be the instance after applying the sequence of $R_5$'s. \cref{claim12} and \cref{sequence-reduction} imply that $\mms^{n'}_{v_i}(M') \geq 1 - 12\delta$.
\end{proof}

For the sake of contradiction, assume that agent $i \in N^1_1$ does not receive a bag at the end of \cref{algo:new}. By Lemma \ref{mms-bound}, $\mms_i \geq 1 - 12\delta$ after applying the sequence of $R_5$'s. By \cref{valid-rules}, $R_2$ and $R_3$ are valid reductions for $i$ and, therefore, $\mms_i \geq 1 - 12\delta$ at the beginning of the bag-filling phase. Let us abuse the notation and assume the instance at this step is $([n],[m],\mathcal{V})$.
\begin{lemma}\label{small-remains-small}
    Assuming $\delta \leq 1/212$, for all $k \in [n]$, if $v_i(C_k) \leq 1-12\delta$, then $v_i(\hat{C}_k) \leq 1-12\delta$.
\end{lemma}
\begin{proof}
    If $\hat{C}_k=C_k$, the claim follows. Otherwise, let $g$ be the last good allocated to $\hat{C}_k$. We have $v_i(\hat{C}_k \setminus g) < 3/4 +\delta$, otherwise $g$ would not be added to $\hat{C}_k$. Since $g > 3n$, by \cref{thm:vr-upper-bounds}, $v_i(g) < 3/16 + \delta/4$. We have
    \begin{align*}
        v_i(\hat{C}_k) &= v_i(\hat{C}_k \setminus g) + v_i(g) \\
        &< \left(\frac{3}{4}+\delta\right) + \left(\frac{3}{16}+\frac{\delta}{4}\right) \\
        &= \frac{15}{16}+\frac{5\delta}{4} \leq 1 - 12\delta. &\tag{$\delta \leq 1/212$}
    \end{align*}
\end{proof}    
\begin{lemma}\label{big-bags}
    If $\delta \leq 1/212$, there exists $k \in [n]$ such that $v_i(C_k) > 1-12\delta$.
\end{lemma}
\begin{proof}
    For the sake of contradiction, assume that for all $k \in [n]$, $v_i(C_k) \leq 1-12\delta$. Since $i$ did not receive a bag at the end of \cref{algo:new}, there exists an unallocated bag $\hat{C}_t$ such that $v_i(\hat{C}_t) < 3/4 +\delta$. We have
    \begin{align*}
        v_i(M) &= \sum_{k \in [n]} v_i(\hat{C}_k) = \sum_{k \neq t} v_i(\hat{C}_k) + v_i(\hat{C}_t) \\
        &< (n-1)(1-12\delta) + (\frac{3}{4}+\delta) &\tag{\cref{small-remains-small} and $v_i(\hat{C}_t) < \frac{3}{4}+\delta$} \\
        &< n(1-12\delta), &\tag{$\delta \leq 1/212$}
    \end{align*}
    Note that $\mms_i \geq 1-12\delta$ and thus $v_i(M) \geq n(1-12\delta)$ which is a contradiction and therefore, \cref{big-bags} holds.
\end{proof}
Let $t$ be largest s.t. $v_i(C_t) > 1-12\delta$. 
\begin{observation}\label{t2}
    Assuming $\delta \leq 1/212$, $t > 1$.
\end{observation}
\begin{proof}
    For the sake of contradiction, assume $t=1$. Since $1-12\delta \ge 3/4 + \delta$, we have 
    \begin{align*}
        v_i(\hat{C}_1) = v_i(C_1) &= v_i(1) + v_i(2n) + v_i(2n+1) \\
        &\leq v_i(1) + v_i(2) + (\frac{1}{4}+\frac{\delta}{3}) &\tag{\cref{thm:vr-upper-bounds}}\\
        &< \left(\frac{3}{4}+\delta\right) + \left(\frac{1}{4}+\frac{\delta}{3}\right) = 1 + \frac{4\delta}{3}. &\tag{$R_5(3/4 +\delta)$ is not applicable}
    \end{align*}
Also, since no bag is allocated to agent $i$, there must be a bag like $C_\ell$ with $v_i(\hat{C}_\ell) < \frac{3}{4}+\delta$.
    \begin{align*}
        n(1-12\delta) \le v_i(M) &= v_i(\hat{C}_1) + \sum_{k \in ([n] \setminus \{1,\ell\})} v_i(\hat{C}_k) + v_i(\hat{C}_\ell) \\
        &< (1+\frac{4\delta}{3}) + (n-2)(1-12\delta) + \frac{3}{4}+\delta, &\tag{\cref{small-remains-small}} \\
        &< n(1-12\delta), &\tag{$\delta \leq 1/212$}
    \end{align*}
    which is a contradiction. Thus, $t>1$.
\end{proof}
\begin{observation}\label{lowerbound-t}
    $v_i(2n+t) > 1/4 -13\delta$.
\end{observation}
\begin{proof}
    We have
\begin{align*}
    1-12\delta < v_i(C_t) &= v_i(t) + v_i(2n-t+1) + v_i(2n+t) \\
    &\leq v_i(1) + v_i(2) + v_i(2n+t) &\tag{$t\geq 1$ and $2n-t+1 \geq 2$}\\
    &< \frac{3}{4} + \delta + v_i(2n+t). &\tag{$R_5$ is not applicable}
\end{align*}
Therefore, $v_i(2n+t) > 1/4 -13\delta$.     
\end{proof}
\begin{observation}\label{ineq-imp}    
    $v_i(2n-t+1) > 3/8 - \delta(12 + 5/6)$.
\end{observation}
\begin{proof}
    Since $R_5$ is not applicable, $v_i(1) + v_i(2) < 3/4 +\delta$ and therefore, $v_i(2) < 3/8 + \delta/2$. We have
    \begin{align*}
        1-12\delta < v_i(C_t) &= v_i(t) + v_i(2n-t+1) + v_i(2n+t) &\tag{$C_t = \{t, 2n-t+1, 2n+t\}$} \\
        &\leq v_i(2) + v_i(2n-t+1) + (\frac{1}{4}+\frac{\delta}{3}) &\tag{$t \geq 2$ by \cref{t2} and $v_i(2n+t) < \frac{1}{4}+\frac{\delta}{3}$ by  \cref{thm:vr-upper-bounds}} \\
        &< (\frac{3}{8}+\frac{\delta}{2}) + v_i(2n-t+1) + (\frac{1}{4}+\frac{\delta}{3}) &\tag{$v_i(2) < \frac{3}{8}+\frac{\delta}{2}$} \\
        &= v_i(2n-t+1) + \frac{5}{8} + \frac{5\delta}{6}.
    \end{align*}
    Therefore, $v_i(2n-t+1) > 3/8 - \delta(12 + 5/6)$.
\end{proof}
Now let $\ell$ be largest such that $v_i(2n+\ell) \geq \delta(26+2/3)$. 
\begin{observation}\label{landt}
    If $\delta \leq 3/476$, then $\ell \geq t$.
\end{observation}
\begin{proof}
    By Observation \ref{lowerbound-t}, $v_i(2n+t) > 1/4 -13\delta$. For $\delta \leq 3/476$, we have $1/4 -13\delta \geq \delta(26+2/3)$. Thus, $\ell \geq t$.
\end{proof}
\begin{lemma}\label{same-bags}
    If $\delta \leq 3/956$, for all $k \leq \min(\ell,n)$, $v_i(C_k) \geq 3/4+\delta$.
\end{lemma}
\begin{proof}
    By \cref{landt}, we have $\ell \geq t$. For all $k \leq t$ we have
    \begin{align*}
        v_i(C_k) &= v_i(k) + v_i(2n-k+1) + v_i(2n+k) &\tag{$C_k = \{k, 2n-k+1, 2n+k\}$}\\
        &\geq v_i(2n-t+1) + 2v_i(2n+t) &\tag{$k \leq 2n-t+1$ and $2n-k+1 < 2n+k \leq 2n+t$}\\
        &> \left(\frac{3}{8}-\delta(12 + \frac{5}{6})\right) + 2\left(\frac{1}{4}-13\delta\right) &\tag{Observation \ref{lowerbound-t} and \ref{ineq-imp}}\\
        &= \frac{7}{8} - \delta(38+\frac{5}{6}) \geq \frac{3}{4}+\delta. &\tag{$\delta \leq 3/956$}
    \end{align*}
    Therefore, no good would be added to $C_k$ for $k \leq t$. Now assume $t<k\leq \ell$. We have
    \begin{align*}
        v_i(C_k) &= v_i(k) + v_i(2n-k+1) + v_i(2n+k) &\tag{$C_k = \{k, 2n-k+1, 2n+k\}$}\\
        &\geq 2v_i(2n-t+1) + v_i(2n+\ell) &\tag{$k < 2n-k+1 < 2n-t+1$ and $2n+k \leq 2n+\ell$}\\
        &> 2\left (\frac{3}{8}-\delta(12 + \frac{5}{6})\right) + \delta(26+\frac{2}{3}) &\tag{\cref{ineq-imp} and the definition of $\ell$}\\
        &= \frac{3}{4}+\delta.
    \end{align*}
\end{proof}
Note that since $i$ does not receive a bag by the end of \cref{algo:new}, there must be a remaining bag $C_k$ such that $v_i(C_k) < 3/4+\delta$. Thus, \cref{same-bags} implies that $\ell < n$ when $\delta \leq 3/956$.
\begin{corollary}[of \cref{same-bags}]\label{same-bags-cor}
    If $\delta \leq 3/956$, for all $k \leq \ell$, $\hat{C}_k = C_k$. 
\end{corollary}
\begin{observation}\label{small-goods-2}
    $v_i(M \setminus \{1,2, \ldots, 2n+\ell\}) \geq (n-\ell)(1/4-13\delta)$.
\end{observation}
\begin{proof}
    Consider the set of goods $\{1,2, \ldots, 2n+\ell\}$ in the \mms~partition of agent $i$. At least $n-\ell$ bags in the \mms~partition have at most two goods in $\{1,2, \ldots, 2n+\ell\}$. Let $P$ be the set of these bags. For all $B \in P$, we have $v_i(B \cap \{1,2, \ldots, 2n+\ell\}) \leq 3/4 +\delta$ since $|B \cap \{1,2, \ldots, 2n+\ell\}| \leq 2$ and $R_5$ is not applicable. Therefore, $v_i(B \setminus \{1,2, \ldots, 2n+\ell\}) \geq (1-12\delta) - (3/4 +\delta) = 1/4 -13\delta$. We have
    \begin{align*}
        v_i(M \setminus \{1,2, \ldots, 2n+\ell\}) &\geq v_i(\cup_{B \in P} B \setminus \{1,2, \ldots, 2n+\ell\}) \\
        &\geq (n-\ell)(\frac{1}{4} - 13\delta).
\qedhere    \end{align*}    
\end{proof}
\begin{lemma}\label{upper-small-goods}
    If $\delta \leq 796$, for all $k>\ell$, $v_i(\hat{C}_k \setminus \{k,2n-k+1\}) < 1/4 -13\delta$.
\end{lemma}
\begin{proof}
    Since $1/4 -13\delta \geq \delta(53+1/3)$ for $\delta \leq 3/796$, it suffices to prove $v_i(\hat{C}_k \setminus \{k,2n-k+1\}) < \delta(53+1/3)$.
    Note that for all $k>\ell$, $v_i(2n+k) < \delta(26+2/3)$. Therefore, if $\hat{C}_k = C_k = \{k,2n-k+1,2n+k\}$, the observation holds. 
    Moreover, we have 
    \begin{align}
        v_i(\{k,2n-k+1\}) &\geq 2v_i(2n-t+1) &\tag{$k < 2n-k+1 \leq 2n-t+1$}\\
        &> 2\left(\frac{3}{8} - \delta(12 + \frac{5}{6})\right) &\tag{\cref{ineq-imp}} \\
        &= \frac{3}{4} - \delta(25 + \frac{2}{3}). \label{ineq20}
    \end{align}
    If $\hat{C}_k \neq C_k$, let $g$ be the last good added to $\hat{C}_k$. Since $g > 3n+1 > 2n+\ell$, $v_i(g) < \delta(26+2/3)$. We have $v_i(\hat{C}_k \setminus g) < 3/4 +\delta$ otherwise $g$ would not be added to $\hat{C}_k$. We have
    \begin{align*}
        v_i(\hat{C}_k) &= v_i(\hat{C}_k \setminus g) + v_i(g)\\
        &< \left(\frac{3}{4}+\delta\right) + \delta(26+\frac{2}{3}) \tag{\cref{thm:vr-upper-bounds}}\\
        &= \frac{3}{4}+\delta(27 + \frac{2}{3}).
    \end{align*}
    Hence,
    \begin{align*}
        \frac{3}{4}+\delta(27 + \frac{2}{3}) &> v_i(\hat{C}_k) \\
        &= v_i(\{k,2n-k+1\}) + v_i(\hat{C}_k \setminus \{k,2n-k+1\}) \\
        &> \frac{3}{4} - \delta(25 + \frac{2}{3}) + v_i(\hat{C}_k \setminus \{k,2n-k+1\}). \tag{Inequality \eqref{ineq20}}
    \end{align*}
    Thus, 
    \begin{align*}
        v_i(\hat{C}_k \setminus \{k,2n-k+1\}) &<\delta(53+\frac{1}{3}).    
    \end{align*}
\end{proof}

We are ready to prove Lemma \ref{n11}.
\begin{lemma}\label{n11}
    For $\delta \leq \lbdelta$, given a $\delta$-ONI instance with $|N^1_1| > \lbnone$, all agents in $N_1^1$ receive a bag of value at least $3/4 +\delta$ at the end of \cref{algo:new}.
\end{lemma}
\begin{proof}
    It suffices to prove that all agents $i \in N^1_1$ receive a bag at the end of \cref{algo:new}. Towards a contradiction, assume that $i \in N^1_1$ does not receive any bag. By \cref{big-bags}, there exists a $k \in [n]$ such that $v_i(C_k) > 1-12\delta$. Recall that $\ell$ is largest such that $v_i(2n+\ell) \geq \delta(26+2/3)$. We have
    \begin{align*}
        (n-\ell)(\frac{1}{4}-13\delta) &\leq v_i(M \setminus \{1,2, \ldots, 2n+\ell\}) &\tag{\cref{small-goods-2}} \\
        &= \sum_{k > \ell} v_i(\hat{C}_k \setminus \{k,2n-k+1\}) &\tag{$\hat{C}_k = C_k$ for $k \in [\ell]$ by \cref{same-bags-cor}} \\
        &< (n-\ell)(\frac{1}{4}-13\delta), &\tag{\cref{upper-small-goods}}
    \end{align*}
    which is a contradiction.
\end{proof}

\begin{theorem}\label{thm:4}
    Given any $\delta \leq \lbdelta$, for all $\delta$-ONI instances where $|N^1_1| > \lbnone$, \cref{algo:new} returns a $(\frac{3}{4}+\delta)$-\mms~allocation.
\end{theorem}
\begin{proof}
    For all other agents $i$, if $i \in N^1_2 \cup N^2$, by \cref{n12andn2}, $i$ receives a bag of value at least $\frac{3}{4}+\delta$ and if $i \in N^1_1$, by \cref{n11} $i$ receives such a bag. Since $N = N^1_1 \cup N^1_2 \cup N^2$, the theorem follows.
\end{proof}

\section{\boldmath $(3/4 + \epsilon)$-MMS allocations}
In this section, we give the complete algorithm $\mathtt{mainApproxMMS}(\mathcal{I}, \alpha)$ that achieves an $\alpha$-MMS allocation for any instance $\I$ with additive valuations and any $\alpha = 3/4 + \epsilon$ for $\epsilon \leq 3/4220$. To this end, first we obtain a $\delta$-ONI instance for $\delta = 4\epsilon/(1-4\epsilon)$ by running $\mathtt{order}(\mathtt{normalize}(\mathtt{reduce}(\I, \epsilon)))$. Then depending on whether $|N^1_1| \leq \lbnone$ or $|N^1_1| > \lbnone$, we run $\mathtt{approxMMS1}$ or $\mathtt{approxMMS2}$. The pseudocode of our algorithm $\mathtt{mainApproxMMS(\I, \alpha)}$ is shown in Algorithm \ref{algo:main}.

\begin{algorithm}[t]
    \caption{$\mathtt{mainApproxMMS}(\mathcal{I}, \alpha)$}
    \label{algo:main}
    \textbf{Input:} Instance $\mathcal{I}=(N,M,\mathcal{V})$ and approximation factor $\alpha > 3/4$
    
    \textbf{Output:} Allocation $A = \langle A_1, \ldots, A_n \rangle$
    \begin{algorithmic}
        \State $\epsilon = \alpha - 3/4$
        \State $\delta = \lbdelta$
        \State $\I = \mathtt{order(normalize(reduce}(\I, \epsilon)))$
        \State $N^1_1 = \{i \in [n] \mid \forall j \in [n]: v_i(B_j) \leq 1 \text{ and } v_i(2n+1) \geq 1/4 - 5\delta\}$
    \If{$|N^1_1| \leq \lbnone$}
        \State \textbf{return} $\mathtt{aprroxMMS1}(\I, \delta)$ \Comment{\cref{algo:gt} in \cref{small-N11}}
    \Else 
        \State \textbf{return} $\mathtt{aprroxMMS2}(\I, \delta)$ \Comment{\cref{algo:new} in \cref{large-N11}}
    \EndIf
    \State \Return $\langle A_1, \ldots, A_n \rangle$
    \end{algorithmic}
\end{algorithm}

\begin{theorem} \label{thm:main}
    Given any instance $\I = (N, M, \mathcal{V})$ where agents have additive valuations and any $\alpha \leq \frac{3}{4} + \lbepsilonfrac$, $\mathtt{mainApproxMMS}(\mathcal{I}, \alpha)$ returns an $\alpha$-MMS allocation for $\I$.
\end{theorem}
\begin{proof}
    Let $\epsilon = \alpha - 3/4$ and $\hat{\I} = \mathtt{order}(\mathtt{normalize}(\mathtt{reduce}(\I, \epsilon)))$. Then by \cref{thm:2}, $\hat{\I}$ is ordered, normalized and $(\frac{3}{4}+\frac{4\epsilon}{1-4\epsilon})$-irreducible ($\frac{4\epsilon}{1-4\epsilon}$-ONI). Since $\epsilon \leq \lbepsilonfrac$, $\frac{4\epsilon}{1-4\epsilon} \leq \lbdeltafrac = \delta$. Thus, $\hat{\I}$ is $\delta$-ONI.
    Furthermore, from any $\beta$-MMS allocation of $\hat{\I}$ one can obtain a $\min(\frac{3}{4}+\epsilon, (1-4\epsilon)\beta)$-MMS allocation of $\I$. 
    
    By \cref{thm:3}, given any $\delta \leq \lbdelta$, for all $\delta$-ONI instances where $|N^1_1| \leq \lbnone$, $\mathtt{approxMMS1}$ returns a $(\frac{3}{4}+\delta)$-\mms~allocation. Also, by \cref{thm:4}, for all $\delta$-ONI instances where $|N^1_1| > \lbnone$, $\mathtt{approxMMS2}$ returns a $(\frac{3}{4}+\delta)$-\mms~allocation. Therefore, $\mathtt{mainApproxMMS}(\mathcal{I}, \alpha)$ returns a $\min(\frac{3}{4}+\epsilon, (1-4\epsilon)(\frac{3}{4} + \delta))$-MMS allocation of $\I$. We have
    \begin{align*}
        (1-4\epsilon)(\frac{3}{4} + \delta) &\geq( 1- \frac{3}{959})(\frac{3}{4} + \lbdeltafrac) \\
        &= \frac{3}{4}+\lbepsilonfrac \\
        &\geq \frac{3}{4} + \epsilon = \alpha.
    \end{align*}
    Thus, $\mathtt{mainApproxMMS}(\mathcal{I}, \alpha)$ returns an $\alpha$-MMS allocation of $\I$.
\end{proof}

\appendix
\section{Missing Proofs}\label{app:mp}
\upper*
\begin{proof}
It suffices to prove $v_i(2n-k+1) \leq 1/3$ and then $v_i(k) > 2/3$ follows.
Let $P = (P_1, \ldots, P_n)$ be an MMS partition of agent $i$.
For $j \in [k]$ and $j' \in [2n+1-k]$,
$v_i(j) + v_i(j') \ge v_i(k) + v_i(2n+1-k) > 1$, since the instance is ordered.
Furthermore, $j$ and $j'$ cannot be in the same bundle in $P$ since the instance is normalized.
In particular, no two goods from $[k]$ are in the same bundle in $P$.
Hence, assume without loss of generality that $j \in P_j$ for all $j \in [k]$.

For all $j \in [k]$ and $j' \in [2n-k+1]$, $j' \not\in P_j$.
Thus, $\{k+1, \ldots, 2n-k+1\} \subseteq P_{k+1} \cup \ldots \cup P_n$.
By pigeonhole principle, there exists a bundle $B \in \{P_{k+1}, \ldots, P_{n}\}$
that contains at least $3$ goods $g_1, g_2, g_3$ in $\{k+1, \ldots, 2n-k+1\}$. Hence,
\begin{align*}
v_i(2n-k+1) &\leq \min_{g \in \{g_1, g_2, g_3\}} v_i(g)
\leq \frac{1}{3}\sum_{g \in \{g_1, g_2, g_3\}} v_i(g)
\leq \frac{v_i(B)}{3} = \frac{1}{3}.\qedhere
\end{align*}
\end{proof}

\trickylemma*
\begin{proof}
    Let $S \in A^+$ be the set of $\ell$ smallest indices in $A^+$ and $L \in A^+$ be the set of $\ell$ largest indices in $A^+$. Since $\hat{B}_k = B_k, \forall k\in A^+$, we have
    $$ \sum_{k \in A^+} v_i(\hat{B}_k) = ({\sum_{k \in S} v_i(k) +  \sum_{k \in L} v_i(2n-k+1)}) + ({\sum_{k \in A^+ \setminus S} v_i(k) + \sum_{k \in A^+ \setminus L} v_i(2n-k+1)}). $$
    We upper bound $({\sum_{k \in S} v_i(k) +  \sum_{k \in L} v_i(2n-k+1}))$ and $({\sum_{k \in A^+ \setminus S} v_i(k) + \sum_{k \in A^+ \setminus L} v_i(2n-k+1}))$ in Claims \ref{claim-1} and \ref{claim-2} respectively.
    \begin{claim}\label{claim-1}
        ${\sum_{k \in S} v_i(k) +  \sum_{k \in L} v_i(2n-k+1) < \ell (\frac{13}{12}+\delta).}$
    \end{claim}
    \begin{claimproof}
    Note that $v_i(k) < 3/4+\delta$ by \cref{thm:vr-upper-bounds} and $v_i(2n-k+1) \leq 1/3$ by \cref{upper-13}. Thus,
    \begin{align*}
        \sum_{k \in S} v_i(k) + \sum_{k \in L} v_i(2n-k+1) < \ell (\frac{3}{4}+\delta + \frac{1}{3}) = \ell(\frac{13}{12}+\delta).
    \end{align*}
    Therefore, \cref{claim-1} holds.
    \end{claimproof}
    \begin{claim}\label{claim-2}
        {$\sum_{k \in A^+ \setminus S} v_i(k) + \sum_{k \in A^+ \setminus L} v_i(2n-k+1) < |A^+| - \ell$.}
    \end{claim}
    \begin{claimproof}
    Assume $A^+ = \{g_1, \ldots, g_{|A^+|}\}$ and $g_1< \ldots< g_{|A^+|}$. Then, $A^+ \setminus S = \{ g_{\ell+1}, \ldots, g_{|A^+|}\}$ and $A^+ \setminus L = \{g_1, \ldots, g_{|A^+|-\ell}\}$. The idea is to pair the goods $g_{k+\ell}$ and $2n- g_k +1$ and prove that their value is less than $1$ for agent $i$. Since $g_{k+\ell} \geq g_k + \ell$, $v_i(g_{k+\ell}) + v_i(2n- g_k +1) < 1$ by the definition of $\ell$. We have
    \begin{align*}
        \sum_{k \in A^+ \setminus S} v_i(k) + \sum_{k \in A^+ \setminus L} v_i(2n-k+1) &= \sum_{k \in [|A^+|-\ell]} (v_i(g_{k+\ell}) + v_i(2n- g_k +1)) 
        < |A^+| - \ell.
    \end{align*}
    Therefore, \cref{claim-2} holds.
    \end{claimproof}

    \cref{claim-1} and \cref{claim-2} together imply \cref{tricky-bound}.
\end{proof}

\bibliographystyle{alpha}
\bibliography{ref}

\end{document}